\documentclass[journal]{IEEEtran}
\IEEEoverridecommandlockouts

\usepackage{amsmath,amsfonts,mathtools,amsthm}
\usepackage[noend]{algorithmic}
\usepackage{algorithm}
\usepackage{array}
\usepackage[caption=false,font=footnotesize]{subfig}
\usepackage{textcomp}
\usepackage{stfloats}
\usepackage{url}
\usepackage{color}
\usepackage{verbatim}
\usepackage{graphicx}
\usepackage{cite}
\usepackage[hidelinks]{hyperref}
\usepackage[nolist]{acronym}
\usepackage{bbm}
\usepackage{xifthen}
\def\BibTeX{{\rm B\kern-.05em{\sc i\kern-.025em b}\kern-.08em
    T\kern-.1667em\lower.7ex\hbox{E}\kern-.125emX}}

\graphicspath{{figures/}}

\newcommand\plotwidth{0.31}

\begin{acronym}
\acro{5G}{fifth generation}
\acro{6G}{sixth generation}
\acro{AMP}{approximate message passing}
\acro{AP}{access point}
\acro{APP}{a posteriori probability}
\acro{AWGN}{additive white Gaussian noise}
\acro{B5G}{beyond fifth generation}
\acro{BG}{Bernoulli-Gaussian}
\acro{BiG-AMP}{bilinear generalized approximate message passing}
\acro{BiGaBP}{bilinear Gaussian belief propagation}
\acro{BP}{belief propagation}
\acro{BPSK}{binary phase-shift-keying}
\acro{CDF}{cumulative distribution function}
\acro{CF-MaMIMO}{cell-free massive multiple-input multiple-output}
\acro{CPU}{central processing unit}
\acro{CSI}{channel state information}
\acro{DER}{detection error rate}
\acro{DFT}{discrete Fourier transform}
\acro{EP}{expectation propagation}
\acro{GaBP}{Gaussian belief propagation}
\acro{GF-CF-MaMIMO}{grant-free cell-free massive multiple-input multiple-output}
\acro{iid}[i.i.d.]{independent and identically distributed}
\acro{IoT}{Internet-of-Things}
\acro{JAC}{joint activity detection and channel estimation}
\acro{JAC-EP}{joint activity detection and channel estimation via expectation propagation}
\acro{JACD}{joint activity detection, channel estimation, and data detection}
\acro{JACD-EP}{joint activity detection, channel estimation, and data detection via expectation propagation}
\acro{JACD-EP-BG}{joint activity detection, channel estimation, and data detection via expectation propagation with Bernoulli-Gaussian distributions}
\acro{JCD}{joint channel estimation and data detection}
\acro{KL}{Kullback-Leibler}
\acro{MaMIMO}{massive multiple-input multiple-output}
\acro{MIMO}{multiple-input multiple-output}
\acro{MAP}{maximum a posteriori}
\acro{MMSE}{minimum mean squared error}
\acro{MMV-AMP}{multiple measurement vector approximate message passing}
\acro{MTC}{machine-type communications}
\acro{NMSE}{normalized mean squared error}
\acro{PC}{pilot contamination}
\acro{PDF}{probability density function}
\acro{PMF}{probability mass function}
\acro{QAM}{quadrature amplitude modulation}
\acro{QoS}{quality of service}
\acro{SER}{symbol error rate}
\acro{UE}{user equipment}
\acro{UPP}{uniform point process}
\end{acronym}

\newcommand{\lvec}[1]{\ensuremath{\mathbf{#1}}}  		
\newcommand{\gvec}[1]{\ensuremath{\boldsymbol{#1}}}		
\newcommand{\lmat}[1]{\ensuremath{\mathbf{#1}}}  		
\newcommand{\gmat}[1]{\ensuremath{\boldsymbol{#1}}}		
\newcommand{\transymb}{\ensuremath{\mathsf{T}}}
\newcommand{\hermsymb}{\ensuremath{\mathsf{H}}}
\newcommand{\tran}{\ensuremath{^{\mkern-1mu\transymb}}}
\newcommand{\herm}{\ensuremath{^{\mkern-1mu\hermsymb}}}
\newcommand{\trace}[1]{\ensuremath{\mathrm{tr}\!\left\{#1\right\}}}
\newcommand{\er}{\mbox{\ensuremath{\mathrm{e}}}}
\newcommand*{\argmin}{\ensuremath{\mathop{\mathrm{arg\,min}}}}
\newcommand*{\argmax}{\ensuremath{\mathop{\mathrm{arg\,max}}}}
\newcommand{\est}[2]{\ensuremath{E_{#1}\!\left\{#2\right\}}}

\newcommand{\Cset}{\mathbb{C}}

\newcommand{\CN}[1]{\ensuremath{\mathcal{CN}\!\left(#1\right)}}

\newcommand{\cat}[1]{\ensuremath{\pi\!\left(#1\right)}}
\newcommand{\BG}[1]{\ensuremath{\mathcal{BG}\!\left(#1\right)}}

\newcommand{\diag}[1]{\ensuremath{\mathrm{diag}\!\left\{#1\right\}}}
\newcommand{\vect}[1]{\ensuremath{\mathrm{vec}\!\left\{#1\right\}}}
\newcommand{\ind}[1]{\ensuremath{\mathbbm{1}_{#1}}}
\newcommand{\pilot}[1]{\ensuremath{#1^{p}}}
\newcommand{\data}[1]{\ensuremath{#1^{d}}}
\newcommand{\msg}[2]{\ensuremath{m_{#1;#2}}}

\newcommand{\msgp}[2]{\ensuremath{p_{#1;#2}}}
\newcommand{\mmd}[2]{\ensuremath{q_{#1;#2}}}
\newcommand{\catmsg}[2]{\ensuremath{\pi_{#1;#2}}}
\newcommand{\Actmsg}[2]{\ensuremath{\lambda_{#1;#2}}}
\newcommand{\Mumsg}[2]{\ensuremath{\gvec{\mu}_{#1;#2}}}
\newcommand{\Cmsg}[2]{\ensuremath{\lvec{C}_{#1;#2}}}
\newcommand{\Kappamsg}[2]{\ensuremath{\kappa_{#1;#2}}}
\newcommand{\Gammamsg}[2]{\ensuremath{\gvec{\gamma}_{#1;#2}}}
\newcommand{\Lambdamsg}[2]{\ensuremath{\gvec{\Lambda}_{#1;#2}}}
\newcommand{\Actmsga}[2]{\ifthenelse{\isempty{#2}}{\ensuremath{\check{\lambda}_{#1}}}{\ensuremath{\check{\lambda}_{#1;#2}}}}
\newcommand{\Mumsga}[2]{\ifthenelse{\isempty{#2}}{\ensuremath{\check{\gvec{\mu}}_{#1}}}{\ensuremath{\check{\gvec{\mu}}_{#1;#2}}}}
\newcommand{\Cmsga}[2]{\ifthenelse{\isempty{#2}}{{\ensuremath{\check{\lvec{C}}_{#1}}}}{\ensuremath{\check{\lvec{C}}_{#1;#2}}}}
\newcommand{\Kappamsga}[2]{\ifthenelse{\isempty{#2}}{\ensuremath{\check{\kappa}_{#1}}}{\ensuremath{\check{\kappa}_{#1;#2}}}}
\newcommand{\Gammamsga}[2]{\ifthenelse{\isempty{#2}}{\ensuremath{\check{\gvec{\gamma}}_{#1}}}{\ensuremath{\check{\gvec{\gamma}}_{#1;#2}}}}
\newcommand{\Lambdamsga}[2]{\ifthenelse{\isempty{#2}}{\ensuremath{\check{\gvec{\Lambda}}_{#1}}}{\ensuremath{\check{\gvec{\Lambda}}_{#1;#2}}}}
\newcommand{\Actmsgb}[2]{\ifthenelse{\isempty{#2}}{\ensuremath{\hat{\lambda}_{#1}}}{\ensuremath{\hat{\lambda}_{#1;#2}}}}
\newcommand{\Mumsgb}[2]{\ifthenelse{\isempty{#2}}{\ensuremath{\hat{\gvec{\mu}}_{#1}}}{\ensuremath{\hat{\gvec{\mu}}_{#1;#2}}}}
\newcommand{\Cmsgb}[2]{\ifthenelse{\isempty{#2}}{{\ensuremath{\hat{\lvec{C}}_{#1}}}}{\ensuremath{\hat{\lvec{C}}_{#1;#2}}}}
\newcommand{\Kappamsgb}[2]{\ifthenelse{\isempty{#2}}{\ensuremath{\hat{\kappa}_{#1}}}{\ensuremath{\hat{\kappa}_{#1;#2}}}}
\newcommand{\Gammamsgb}[2]{\ifthenelse{\isempty{#2}}{\ensuremath{\hat{\gvec{\gamma}}_{#1}}}{\ensuremath{\hat{\gvec{\gamma}}_{#1;#2}}}}
\newcommand{\Lambdamsgb}[2]{\ifthenelse{\isempty{#2}}{\ensuremath{\hat{\gvec{\Lambda}}_{#1}}}{\ensuremath{\hat{\gvec{\Lambda}}_{#1;#2}}}}
\newcommand{\mumsg}[2]{\ensuremath{\mu_{#1;#2}}}

\newtheorem{proposition}{Proposition}
\newtheorem{corollary}{Corollary}

\begin{document}
\bstctlcite{IEEEexample:BSTcontrol}

\title{Expectation Propagation for Distributed Inference in Grant-Free Cell-Free Massive MIMO%
\thanks{This work was funded by the Deutsche Forschungsgemeinschaft (DFG, German Research Foundation) – Project CO 1311/1-1,  Project ID 491320625.
An earlier version of this paper was presented in part at the IEEE International Workshop on Signal Processing Advances in Wireless Communications (SPAWC), 2024~\cite{Forsch2024}.
}
\thanks{Christian Forsch and Laura Cottatellucci are with the Institute for Digital Communications, Friedrich-Alexander-Universität Erlangen-Nürnberg, Erlangen, Germany (e-mail: christian.forsch@fau.de; laura.cottatellucci@fau.de).}
}

\author{Christian~Forsch,~\IEEEmembership{Graduate Student Member,~IEEE,} and~Laura~Cottatellucci,~\IEEEmembership{Member,~IEEE}}

\maketitle

\begin{abstract}
\Ac{GF-CF-MaMIMO} systems are anticipated to be a key enabling technology for next-generation \ac{IoT} networks, as they support massive connectivity without explicit scheduling.
However, the large amount of connected devices prevents the use of orthogonal pilot sequences, resulting in severe \ac{PC} that degrades channel estimation and data detection performance.
Furthermore, scalable \ac{GF-CF-MaMIMO} networks inherently rely on distributed signal processing.
In this work, we consider the uplink of a \ac{GF-CF-MaMIMO} system and propose two novel \emph{distributed} algorithms for \emph{\ac{JACD}} based on \ac{EP}.
The first algorithm, denoted as \acs{JACD-EP}\acused{JACD-EP}, uses Gaussian approximations for the channel variables, whereas the second, referred to as \acs{JACD-EP-BG}\acused{JACD-EP-BG}, models them as \ac{BG} random variables.
To integrate the \ac{BG} distribution into the \ac{EP} framework, we derive its exponential family representation and develop the two algorithms as efficient message passing over a factor graph constructed from the \ac{APP} distribution.
The proposed framework is inherently scalable with respect to both the number of \acp{AP} and \acp{UE}.
Simulation results show the efficient mitigation of \ac{PC} by the proposed distributed algorithms and their superior detection accuracy compared to (genie-aided) centralized linear detectors.
\end{abstract}

\begin{IEEEkeywords}
Expectation propagation, Bernoulli-Gaussian distribution, distributed inference, activity detection, channel estimation, data detection, grant-free cell-free massive MIMO.
\end{IEEEkeywords}

\acresetall

\vspace*{-3mm}
\section{Introduction}\label{sec:intro}
\IEEEPARstart{T}{he} explosive growth of \ac{IoT} devices and the evolving requirements of massive and critical \ac{MTC} for ultra-low latency, high reliability, and energy efficiency call for new communication system designs~\cite{Mahmood2021}.
Traditional grant-based multiple access schemes fail to meet these high demands due to signaling overhead and latency, especially under dense device deployments and sporadic burst traffic.
Hence, grant-free random access has emerged as a compelling alternative, facilitating efficient resource utilization without the need for an explicit scheduling grant~\cite{Liu2018b,Shahab2020,Gao2024}.
Simultaneously, the \ac{CF-MaMIMO}\acused{MaMIMO}\acused{MIMO} network architecture, in which a large number of geographically distributed \acp{AP} jointly serve a potentially large number of \acp{UE}, enables ubiquitous coverage and energy-efficient communication~\cite{Ngo2017,Ngo2018,Mohammadi2024}.
Grant-free \ac{CF-MaMIMO} (\acs{GF-CF-MaMIMO})\acused{GF-CF-MaMIMO} systems combine the advantages of grant-free random access and \ac{CF-MaMIMO} and are promising candidates for next-generation \ac{IoT} networks~\cite{Wang2021,Ganesan2021}.
At the same time, \ac{CF-MaMIMO} architectures implemented with centralized processing face fundamental scalability limitations~\cite{Xu2025} which are further exacerbated in \ac{GF-CF-MaMIMO} systems, where a large number of sporadically active \acp{UE} must be jointly processed, motivating the design of distributed algorithms.

In addition to the classical communication tasks of channel estimation and data detection, grant-free random access requires identifying the set of active \acp{UE} transmitting data.
To accomplish this, \acp{UE} transmit pilot sequences which can be used to jointly estimate both device activities and channels, a process known as \ac{JAC}, e.g.,~\cite{Liu2018a,Ke2020}.
However, the large number of connected devices in massive \ac{MTC} prevents the use of orthogonal pilot sequences, leading to \emph{\ac{PC}} which degrades overall system performance and further complicates activity and data detection.
\Ac{PC} occurs in any communication system where \acp{UE} transmit non-orthogonal pilot sequences.
In centralized \ac{MaMIMO} systems, \emph{channel hardening} and \emph{favorable propagation} have been used to alleviate the \ac{PC} problem for grant-based multiple access in~\cite{Ngo2012,Yin2013,Cottatellucci2013,Mueller2014,Yin2016}.
In contrast, such properties typically do not hold in \ac{CF-MaMIMO}~\cite{Yin2014,Chen2018,Gholami2020a,Gholami2020b}, rendering such decontamination methods ineffective.
The use of centralized \ac{MMSE} processing in grant-based centralized \ac{MaMIMO} was shown to mitigate \ac{PC} under practical channel statistics in the asymptotic regime of infinitely many antennas~\cite{Bjoernson2018,Bjoernson2020}.
Similar results were obtained for distributed processing in~\ac{CF-MaMIMO}~\cite{Polegre2021}.
However, \ac{PC} still remains a major practical challenge~\cite{Xu2025}, especially in scalable \ac{CF-MaMIMO} systems, which inherently rely on distributed signal processing, with a limited number of \acp{AP} and antennas per \ac{AP}.
This problem becomes even more severe in \ac{GF-CF-MaMIMO} where the additional need for \ac{UE} activity detection arises and a much larger number of \acp{UE} are simultaneously processed, thereby intensifying \ac{PC}.
These challenges motivate the development of distributed pilot decontamination schemes that can effectively mitigate \ac{PC} while jointly supporting scalability and low fronthaul traffic in \ac{GF-CF-MaMIMO} architectures.

One promising approach to mitigate \ac{PC} is to exploit not only the pilot symbols for channel estimation and \ac{UE} activity detection but also the received data symbols.
In grant-based multiple access, where the \ac{UE} activities are known beforehand and only channels and data symbols need to be estimated, this approach is referred to as \ac{JCD}.
Several studies have explored \ac{JCD} in \ac{CF-MaMIMO} systems.
The authors in~\cite{Gholami2021a} investigated semi-blind methods for \ac{JCD} in \ac{CF-MaMIMO} and derived conditions for semi-blind identifiability.
In~\cite{Song2022}, a \ac{JCD} scheme based on forward-backward splitting was developed, exploiting the sparsity of \ac{CF-MaMIMO} channels and employing non-orthogonal pilot sequences.
A distributed \ac{EP}-based semi-blind \ac{JCD} algorithm for \ac{CF-MaMIMO} systems was presented in~\cite{Karataev2024} and was further refined and analyzed under \ac{PC} in~\cite{Forsch2025}.
The authors in~\cite{Zhao2024} combined variational Bayes and \ac{EP} to develop a semi-blind \ac{JCD} algorithm based on Bethe free energy optimization.

In grant-free random access, the received data symbols can also be used to additionally enhance the \ac{UE} activity detection, a task referred to as \emph{\ac{JACD}}.
In the literature, different \ac{JACD} schemes have been proposed.
One class of algorithms is based on sequences spreading data whereby the data symbols of each \ac{UE} are multiplied by a unique spreading signature, e.g.,~\cite{Jiang2020,Zhang2020,Zhang2025}.
These unique, generally non-orthogonal signatures spread the transmitted symbols in the time/frequency domain, enabling multiple \acp{UE} to share the same resources.
However, spreading data reduces the spectral efficiency and limits the achievable data rate per \ac{UE}.
Hence, in the following, we focus on \ac{JACD} schemes that do not rely on spreading data sequences.
For centralized \ac{MaMIMO} systems, several non-spread \ac{JACD} schemes based on \ac{BiG-AMP} have been proposed.
These include approaches that combine \ac{BiG-AMP} with loopy \ac{BP}~\cite{Zou2020}, introduce vector nodes with correlated channels~\cite{Zhang2023}, or iteratively exchange information with the channel decoder~\cite{Bian2023}.
In~\cite{Iimori2021}, bilinear Gaussian \ac{BP} (\acs{BiGaBP}\acused{BiGaBP}) was applied to the \ac{JACD} problem in \ac{GF-CF-MaMIMO} networks, employing low-coherence pilot sequences.
Here, the beliefs of the \ac{BG}-distributed channels and the categorically distributed data symbols are approximated by Gaussian distributions, and a corresponding message-passing algorithm was derived.
The authors in \cite{Sun2025} developed two \ac{JACD} algorithms based on forward-backward splitting and deep unfolding for hyperparameter optimization.
The proposed centralized algorithms employ Laplace distributions to model the sparsity of channels and data.

In this paper, we propose two novel distributed \ac{JACD} algorithms for \ac{GF-CF-MaMIMO} networks.
We formulate the \ac{JACD} task as a \ac{MAP} estimation and detection problem and, then, solve it approximately using \ac{EP}, a Bayesian learning technique that iteratively computes tractable approximations of factorized probability distributions employing exponential family distributions~\cite{Minka2001a,Minka2001b}.
To this end, we factorize the \ac{APP} distribution of user activities, channels, and data symbols to enable a tractable joint inference on a factor graph based on \ac{EP}.
This factor graph approach enables an inherent distributed implementation of the proposed algorithms which is suitable for \emph{decentralized} and \emph{scalable} \ac{CF-MaMIMO} systems with baseband-processing capabilities at the \acp{AP}.
Part of this work was presented in~\cite{Forsch2024}.
The main novelty of the present paper compared to~\cite{Forsch2024} lies in the adoption of \ac{BG} distributions within the \ac{EP} framework to more accurately capture the sparsity of the effective \ac{UE} channels, thereby significantly enhancing the \ac{JACD} performance.
In contrast, the approach in~\cite{Forsch2024} uses Gaussian approximations for the effective \ac{UE} channels.
To enable \ac{EP} inference with \ac{BG} random variables, we first derive the exponential family representation of the \ac{BG} distribution and then apply \ac{EP} message passing to develop the proposed algorithms.
Finally, we conduct an extensive performance analysis via Monte Carlo simulations and compare the proposed algorithms with optimal linear and state-of-the-art \ac{EP}-based algorithms.
The contributions of this paper are summarized as follows:
\begin{itemize}
    \item We formulate the \ac{JACD} problem as a \ac{MAP} estimation and detection problem and propose a factor graph that enables \ac{EP}-based inference of the approximate \ac{APP} distributions of \ac{UE} activities, channels, and data symbols.
    The proposed formulation yields a natural task partition between \acp{AP} and \ac{CPU}, leading to a distributed algorithm that keeps the computational load at the \ac{CPU} linear in the number of data signals, supports scalability, and facilitates a structured and efficient design of the information exchanged over the fronthaul.
    \item We present the exponential family representation of the \ac{BG} distribution, enabling simple and closed-form multiplication and division of \ac{BG} distributions and their seamless integration into the \ac{EP} framework.
    The \ac{BG} model provides a more effective representation of sparse random variables such as channels between \acp{UE} with unknown activity state and \acp{AP} compared to Gaussian approximations.
    \item We develop two distributed \ac{JACD} algorithms by applying \ac{EP} on factor graphs with Gaussian and \ac{BG} channel models, yielding the proposed \acs{JACD-EP}\acused{JACD-EP} and \acs{JACD-EP-BG}\acused{JACD-EP-BG} algorithm, respectively.
    Both the algorithms exhibit a polynomial computational complexity and enable scalable signal processing in \ac{GF-CF-MaMIMO} systems.
    \item We present an extended numerical analysis of the proposed algorithms' performance for different pilot sequence lengths, modeling different levels of \ac{PC}.
    As benchmark, we adopt centralized linear \ac{MMSE} processing, including genie-aided variants, and show the superior performance of the proposed distributed algorithms, particularly under severe \ac{PC}.
    Since centralized \ac{MMSE} processing provides an upper bound on the performance of distributed linear \ac{MMSE} schemes, no distributed linear \ac{MMSE} benchmark is included.
    The results show that linear \ac{MMSE} processing is not sufficient to effectively combat \ac{PC}, highlighting the need for nonlinear estimation and detection schemes.
\end{itemize}

The remainder of this paper is organized as follows.
In Section~\ref{sec:sys_mod} and~\ref{sec:problem}, we introduce the \ac{GF-CF-MaMIMO} system model and formulate the inference problem, respectively.
In Section~\ref{sec:BG}, we derive the exponential family representation of the \ac{BG} distribution.
Then, in Section~\ref{sec:JACD-EP-BG}, we propose the \ac{JACD-EP} and \ac{JACD-EP-BG} algorithms.
In Section~\ref{sec:sims}, the performance of the proposed algorithms is evaluated via Monte Carlo simulations.
Finally, conclusions are drawn in Section~\ref{sec:concl}.

\textit{Notation:}
Lower case, bold lower case, and bold upper case letters, e.g., $x,\lvec{x},\lmat{X}$, represent scalars, vectors, and matrices, respectively.
$\lmat{I}_N$ is the $N$-dimensional identity matrix.
$\diag{\cdot}$ denotes a diagonal matrix whose main diagonal entries are given by the elements inside the brackets.
$\delta(\cdot)$ represents the Dirac delta function.
The indicator function $\ind{(\cdot)}$ equals one if the condition in the subscript is satisfied and zero otherwise.
$(\cdot)\tran$ and $(\cdot)\herm$ denote the transpose and conjugate transpose (Hermitian) operation, respectively.
The trace of a matrix $\lmat{X}$ is written as $\trace{\lmat{X}}$.
$|\mathcal{S}|$ stands for the cardinality of the set $\mathcal{S}$.
$\est{}{\cdot}$ denotes the expectation operator.
$\CN{\lvec{x}|\gvec{\mu},\lmat{C}}=|\pi\lmat{C}|^{-1} \er^{-(\lvec{x}-\gvec{\mu})\herm\lmat{C}^{-1}(\lvec{x}-\gvec{\mu})}$ represents the \ac{PDF} of a proper complex-valued Gaussian random vector $\lvec{x}$ with mean $\gvec{\mu}$ and covariance matrix $\lmat{C}$.
$\cat{x}$ denotes the \ac{PMF} of a categorical random variable $x$.
The notation $x\sim p$ indicates that the random variable $x$ follows the distribution $p$.
In a factor graph, the message sent from the factor node $\Psi_\alpha$ to the variable node $\lvec{x}_\beta$ is denoted as $\msg{\Psi_\alpha}{\lvec{x}_\beta}$ and consists of parameters of the distribution $\msgp{\Psi_\alpha}{\lvec{x}_\beta}(\lvec{x}_\beta)$ in the exponential family.
Note that we adopt the same subscript also for the parameters, e.g., $\msgp{\Psi_\alpha}{\lvec{x}_\beta}(\lvec{x}_\beta)=\mathcal{CN}\big(\lvec{x}_\beta|\Mumsg{\Psi_\alpha}{\lvec{x}_\beta},\Cmsg{\Psi_\alpha}{\lvec{x}_\beta})$ with mean $\Mumsg{\Psi_\alpha}{\lvec{x}_\beta}$ and covariance matrix $\Cmsg{\Psi_\alpha}{\lvec{x}_\beta}$ or $\msgp{\Psi_\alpha}{\lvec{x}_\beta}(\lvec{x}_\beta)=\catmsg{\Psi_\alpha}{\lvec{x}_\beta}(\lvec{x}_\beta)$ with probability values $\catmsg{\Psi_\alpha}{\lvec{x}_\beta}(\lvec{x}_\beta\big)$ for a Gaussian or categorical random variable $\lvec{x}_\beta$, respectively.
Analogous notation holds for variable-to-factor messages $\msg{\lvec{x}_\beta}{\Psi_\alpha}$.

\vspace*{-1mm}
\section{System Model}\label{sec:sys_mod}
\vspace*{-1mm}
We consider the uplink of a \ac{GF-CF-MaMIMO} network comprising $L$ geographically distributed \acp{AP}, each equipped with $N$ antennas, serving $K$ synchronized single-antenna \acp{UE} as illustrated in Fig.~\ref{fig:system_model}.
\begin{figure}[t]
    \centerline{\includegraphics[width=0.22\textwidth]{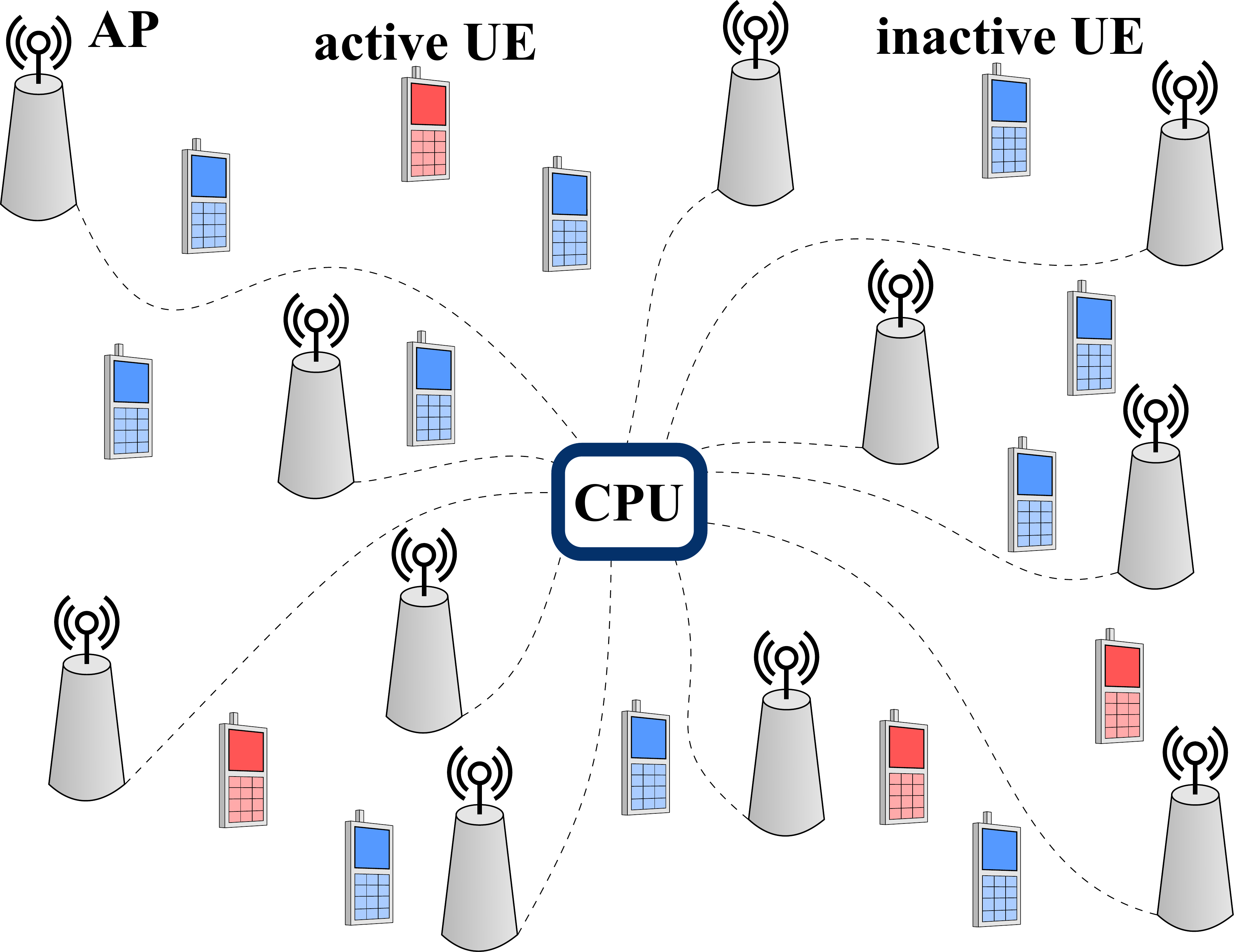}}
    \caption{\Ac{GF-CF-MaMIMO} network with geographically distributed \acp{AP} and \acp{UE} exhibiting different activity states.}
    \label{fig:system_model}
    \vspace*{-2mm}
\end{figure}
All \acp{AP} are connected to a \ac{CPU} via fronthaul links.
Due to the sporadic traffic characteristics of the network, only a subset of the $K$ \acp{UE} is active and transmits data simultaneously.
Let $\lmat{Y}_l=[\lvec{y}_{l,1}\cdots\lvec{y}_{l,T}]\in\Cset^{N\times T}$ denote the signal received at \ac{AP} $l$ over a coherence block of $T$ channel uses.
It is given by
\vspace*{-2mm}
\begin{equation}
	\lmat{Y}_l = \lmat{H}_l\lmat{U}\lmat{X} + \lmat{N}_l = \sum_{k=1}^K\lvec{h}_{l,k}u_k\lvec{x}_k\tran + \lmat{N}_l,
	\label{eq:Y_l}
\end{equation}
where  $\lmat{H}_l=[\lvec{h}_{l,1}\cdots\lvec{h}_{l,K}]\!\in\!\Cset^{N\times K}$ denotes the channel matrix of \ac{AP} $l$ with $\lvec{h}_{l,k}\!\in\!\Cset^{N\times1}$ being the channel between \ac{AP} $l$ and \ac{UE} $k$.
The diagonal activity matrix $\lmat{U}=\diag{u_1,\dots,u_K}\in\{0,1\}^{K\times K}$ contains the binary activity indicators where $u_k=1$ if \ac{UE} $k$ is active and $u_k=0$ otherwise.
$\lmat{X}=[\lvec{x}_1\cdots\lvec{x}_K]\tran\!\in\!\Cset^{K\times T}$ is the transmit symbol matrix with $\lvec{x}_k\!\in\!\Cset^{T\times1}$ representing the transmit sequence of \ac{UE} $k$ and $\lmat{N}_l\in\Cset^{N\times T}$ is the matrix of \ac{iid} \ac{AWGN} elements $n\sim\CN{n|0,\sigma_n^2}$.
The channels and the user activity indicators are constant during the channel coherence interval.
We assume block Rayleigh fading channels, i.e., $\lvec{h}_{l,k}\sim p_{h_{l,k}}(\lvec{h}_{l,k})=\CN{\lvec{h}_{l,k}|\lvec{0}_N,\gmat{\Xi}_{l,k}}$, where $\gmat{\Xi}_{l,k}\in\Cset^{N\times N}$ is the spatial correlation matrix and $\xi_{l,k}=\frac{1}{N}\trace{\gmat{\Xi}_{l,k}}$ is the associated large-scale fading coefficient.
The activity indicator $u_k$ is drawn from a Bernoulli distribution, $u_k\sim p_u(u_k)=(1-\lambda)\ind{u_k=0}+\lambda\ind{u_k=1}$ where $\lambda$ denotes the user activity probability.
The transmit symbol matrix consists of a pilot part $\pilot{\lmat{X}}\in\Cset^{K\times T_p}$ with known pilot symbols $\pilot{x}_{kt}$ and a data part $\data{\lmat{X}}\in\mathcal{X}^{K\times T_d}$ with data symbols $\data{x}_{kt}$ such that $\lmat{X}=[\pilot{\lmat{X}}\;\data{\lmat{X}}]$ and $T_p+T_d=T$.
The transmit symbols belong to the constellation $\mathcal{X}$ of cardinality $M=|\mathcal{X}|$ which does not contain the zero-symbol, i.e., $0\notin\mathcal{X}$.
All \acp{UE} transmit with the same average transmit power $\sigma_x^2=\est{}{|x_{kt}|^2}$.
A similar decomposition holds for the receive matrix, i.e., $\lmat{Y}_l=[\pilot{\lmat{Y}}_l\;\data{\lmat{Y}}_l]$ with received pilots $\pilot{\lmat{Y}}_l\in\Cset^{N\times T_p}$ and received data $\data{\lmat{Y}}_l\in\Cset^{N\times T_d}$.
Finally, we assume $T_p<K$ since the number of \acp{UE} can be very large and assigning orthogonal pilot sequences to the \acp{UE} is generally impractical.
This inevitably leads to the so-called \ac{PC} effect.

\vspace*{-1mm}
\section{Problem Formulation}\label{sec:problem}
\vspace*{-1mm}
The non-orthogonality of the pilot sequences degrades both user activity detection and channel estimation performance.
The detrimental effect of \ac{PC} can be mitigated by exploiting the detected data symbols to improve activity detection and channel estimation, which in turn enables a refinement of the detected data symbols and yields an iterative \ac{JACD} approach.
At first, we summarize the received signals from all \acp{AP} in the global model
\begin{equation}
	\lmat{Y} = \lmat{H}\lmat{U}\lmat{X} + \lmat{N},
	\label{eq:Y}
\end{equation}
where, for $\lmat{A} \equiv \{\lmat{Y}, \lmat{H}, \lmat{N} \}$, we define $\lmat{A}=[\lmat{A}_1\tran\cdots\lmat{A}_L\tran]\tran$.
For \ac{JACD}, the receiver jointly estimates the user-activity, channel, and data matrices $\lmat{U}$,  $\lmat{H}$, and $\data{\lmat{X}}$, respectively.
The \ac{MAP} estimator is given by
\begin{equation}
	(\hat{\lmat{U}},\hat{\lmat{H}},\data{\hat{\lmat{X}}}) = \argmax_{\lmat{U},\lmat{H},\data{\lmat{X}}}\;p(\lmat{U},\lmat{H},\data{\lmat{X}}|\lmat{Y},\pilot{\lmat{X}}),
	\label{eq:MAP}
\end{equation}
where the \ac{APP} distribution $p(\lmat{U},\lmat{H},\data{\lmat{X}}|\lmat{Y},\pilot{\lmat{X}})$ can be factorized using Bayes' rule as
\begin{equation}
	p(\lmat{U},\lmat{H},\data{\lmat{X}}|\lmat{Y},\pilot{\lmat{X}})\propto p(\lmat{Y}|\lmat{U},\lmat{H},\lmat{X})\cdot p(\lmat{U})\cdot p(\lmat{H})\cdot p(\lmat{X}).
	\label{eq:APP}
\end{equation}
Direct \ac{MAP} inference in~\eqref{eq:MAP} is computationally intractable due to the high dimensionality of the involved variables.
Hence, in Section~\ref{sec:JACD-EP-BG}, we propose two low-complexity approximate Bayesian learning methods for \ac{JACD}.
Since these methods rely on \ac{EP} and \ac{BG} distributions, we review the key properties of the exponential family and provide an exponential family representation of the \ac{BG} distribution in the following section.

\vspace*{-2mm}
\section{Bernoulli-Gaussian Distribution in Exponential Family Form}\label{sec:BG}
\vspace*{-1mm}
In this section, we derive the exponential family representation of the \ac{BG} distribution.
To this end, we first recall the general form of an exponential family distribution and, then, show that the \ac{BG} distribution can be expressed in exponential family form.

\vspace*{-2mm}
\subsection{Exponential Family}\label{subsec:exp_fam}
A probability distribution belongs to the exponential family if it can be expressed as~\cite{Wainwright2007}
\begin{equation}
	p(\lvec{x}) = \er^{\gvec{\eta}\herm\lvec{u}(\lvec{x})-A(\gvec{\eta})},
	\label{eq:exp_fam}
\end{equation}
where $\gvec{\eta}$ is the vector of natural parameters, $\lvec{u}(\lvec{x})$ is the vector of sufficient statistics, and $A(\gvec{\eta})$ is the log-partition function.
Exponential family distributions enjoy convenient properties, e.g., they allow simple multiplications and divisions of probability distributions, which makes them attractive in the \ac{EP} framework.
Two prominent members of the exponential family are the Bernoulli and Gaussian distributions reviewed below.

Let $x\in\{0,1\}$ be a Bernoulli-distributed random variable that equals one with probability $\lambda$ and zero with probability $1-\lambda$.
The corresponding \ac{PMF} is given by
\begin{equation}
	p_\text{B}(x) = (1-\lambda)^{\ind{x=0}}\cdot\lambda^{1-\ind{x=0}} = \er^{\eta_\text{B}u_\text{B}(x) - A_\text{B}(\eta_\text{B})},
	\label{eq:Bern}
\end{equation}
with natural parameter $\eta_\text{B}=\log\frac{1-\lambda}{\lambda}$, sufficient statistic $u_\text{B}(x)=\ind{x=0}$, and log-partition function $A_\text{B}(\eta_\text{B})=-\log\lambda=\log(1+\er^{\eta_\text{B}})$.

The exponential family representation of the proper multivariate complex Gaussian distribution with mean $\gvec{\mu}$ and covariance matrix $\lmat{C}$ is given by the \ac{PDF}
\begin{equation}
	p_\text{G}(\lvec{x}) = \CN{\lvec{x}|\gvec{\mu},\lmat{C}} = \er^{\gvec{\eta}_\text{G}\herm\lvec{u}_\text{G}(\lvec{x}) - A_\text{G}(\gvec{\eta}_\text{G})},
	\label{eq:Gauss_exp_fam}
\end{equation}
with natural parameters $\gmat{\Lambda}=\lmat{C}^{-1}$, $\gvec{\gamma}=\lmat{C}^{-1}\gvec{\mu}$, and corresponding vector of natural parameters $\gvec{\eta}_\text{G}=\big[\gvec{\gamma}\tran,\gvec{\gamma}\herm,-\vect{\gmat{\Lambda}}\tran\big]\tran$, sufficient statistics $\lvec{u}_\text{G}(\lvec{x})=\big[\lvec{x}\tran,\lvec{x}\herm,\vect{\lvec{x}\lvec{x}\herm}\tran\big]\tran$, and log-partition function $A_\text{G}(\gvec{\eta}_\text{G})=\gvec{\gamma}\herm\gmat{\Lambda}^{-1}\gvec{\gamma}-\log|\pi^{-1}\gmat{\Lambda}|$.

\vspace*{-2mm}
\subsection{Bernoulli-Gaussian Distribution}\label{subsec:BG_exp_fam}
A \ac{BG} random variable describes two mutually exclusive events.
The first event occurs with probability $\lambda$ and yields a Gaussian random vector with mean $\gvec{\mu}$ and covariance matrix $\lmat{C}$.
In the complementary event, which occurs with probability $1-\lambda$, the random vector is zero.
We denote the probability $\lambda$ as activity probability and the probability mass at zero $1-\lambda$ as inactivity probability.
By introducing the Bernoulli indicator as for the Bernoulli distribution in~\eqref{eq:Bern}, the \ac{BG} model can be expressed in exponential family form as shown by the following proposition.
\begin{proposition}\label{prop:BG_exp_fam}
The exponential family representation of the \ac{BG} distribution with activity probability $\lambda$ and the proper complex Gaussian event characterized by mean $\gvec{\mu}$ and covariance matrix $\lmat{C}$ is given by
\begin{equation}
	p_\mathrm{BG}(\lvec{x}) = \BG{\lvec{x}|\lambda,\gvec{\mu},\lmat{C}} = \er^{\gvec{\eta}_\mathrm{BG}\herm\lvec{u}_\mathrm{BG}(\lvec{x}) - A_\mathrm{BG}(\gvec{\eta}_\mathrm{BG})},
	\label{eq:BG_exp_fam}
\end{equation}
with vector of natural parameters $\gvec{\eta}_\mathrm{BG}=\left[\kappa,\gvec{\eta}_\mathrm{G}\tran\right]\tran$, sufficient statistics $\lvec{u}_\mathrm{BG}(\lvec{x})=\left[\ind{\lvec{x}=\lvec{0}},\lvec{u}_\mathrm{G}(\lvec{x})\tran\right]\tran$, and log-partition function $A_\mathrm{BG}(\gvec{\eta}_\mathrm{BG})=\log\left(\er^{A_\mathrm{G}(\gvec{\eta}_\mathrm{G})}+\er^\kappa\right)$, where $\kappa\coloneq\log\frac{1-\lambda}{\lambda}+A_\mathrm{G}(\gvec{\eta}_\mathrm{G})$.
\end{proposition}
\begin{proof}
The \ac{BG} mixture model can be expressed in terms of the Bernoulli indicator as
\begin{equation*}
	p_\text{BG}(\lvec{x}) = (1-\lambda)^{\ind{\lvec{x}=\lvec{0}}}\cdot(\lambda\cdot\er^{\gvec{\eta}_\text{G}\herm\lvec{u}_\text{G}(\lvec{x}) - A_\text{G}(\gvec{\eta}_\text{G})})^{1-\ind{\lvec{x}=\lvec{0}}}.
\end{equation*}
Collecting all terms into natural parameter and sufficient statistic vectors yields
\begin{align*}
    p_\text{BG}(\lvec{x}) &= \Big(\frac{1-\lambda}{\lambda}\Big)^{\ind{\lvec{x}=\lvec{0}}}\!\!\!\cdot\!\er^{A_\text{G}(\gvec{\eta}_\text{G})\!\cdot\!\ind{\lvec{x}=\lvec{0}}}\!\cdot\!\er^{\log\lambda}\!\cdot\!\er^{\gvec{\eta}_\text{G}\herm\lvec{u}_\text{G}(\lvec{x}) - A_\text{G}(\gvec{\eta}_\text{G})}\\
    &= \er^{\kappa\cdot\ind{\lvec{x}=\lvec{0}} + \gvec{\eta}_\text{G}\herm\lvec{u}_\text{G}(\lvec{x}) + \log\lambda - A_\text{G}(\gvec{\eta}_\text{G})},
\end{align*}
which proves the claim with $\gvec{\eta}_\text{BG}$, $\lvec{u}_\text{BG}(\lvec{x})$, and $A_\text{BG}(\gvec{\eta}_\text{BG})$ defined above.
\end{proof}
We observe that the exponential family representation of the \ac{BG} model introduces an additional natural parameter and sufficient statistic beyond those of the Gaussian distribution which are related to the inactivity event.
The probability distribution $\BG{\lvec{x}|\lambda,\gvec{\mu},\lmat{C}}$ returns the inactivity probability $1\!-\!\lambda$ when $\lvec{x}\hspace{-1em/24}=\hspace{-1em/24}\lvec{0}$, and for $\lvec{x}\hspace{-1em/24}\neq\hspace{-1em/24}\lvec{0}$ the corresponding \ac{PDF} value of the Gaussian event scaled by the activity probability $\lambda$.

The key advantage of the exponential family representation is the straightforward computation of normalized products or quotients of the respective distributions by simple addition or subtraction of the natural parameters, respectively.
For example, the product of two \ac{BG} models characterized by the natural parameters $(\kappa_1,\gvec{\gamma}_1,\gmat{\Lambda}_1)$ and $(\kappa_2,\gvec{\gamma}_2,\gmat{\Lambda}_2)$, respectively, yields another \ac{BG} model with natural parameters $(\kappa_1+\kappa_2,\gvec{\gamma}_1+\gvec{\gamma}_2,\gmat{\Lambda}_1+\gmat{\Lambda}_2)$.
To offer further insights, we present in the following the results in the original parameter space, i.e., activity probability $\lambda$, mean $\gvec{\mu}$, and covariance matrix $\lmat{C}$, and explicitly state the normalization constant.
Similar results can be derived for quotients of \ac{BG} models but are omitted for brevity.
\begin{corollary}[Bernoulli-Gaussian Product Lemma]\label{co:BG_product}
The product of two \ac{BG} models yields a new unnormalized \ac{BG} model,
\begin{align}
\begin{alignedat}{2}
	&\BG{\lvec{x}|\lambda_1,\gvec{\mu}_1,\lmat{C}_1} \,\cdot &&\BG{\lvec{x}|\lambda_2,\gvec{\mu}_2,\lmat{C}_2}\\
    &\quad= \BG{\lvec{x}|\lambda,\gvec{\mu},\lmat{C}} &&\cdot \big[\lambda_1\lambda_2\CN{\lvec{0}|\gvec{\mu}_1-\gvec{\mu}_2,\lmat{C}_1+\lmat{C}_2}\\
	&&&\quad+(1-\lambda_1)(1-\lambda_2)\big]
	\label{eq:BG_product}
\end{alignedat}
\end{align}
with
\vspace*{-2mm}
\begin{align}
	\lambda &= \frac{\lambda_1\lambda_2\CN{\lvec{0}|\gvec{\mu}_1-\gvec{\mu}_2,\lmat{C}_1+\lmat{C}_2}}{\lambda_1\lambda_2\CN{\lvec{0}|\gvec{\mu}_1-\gvec{\mu}_2,\lmat{C}_1+\lmat{C}_2}+(1-\lambda_1)(1-\lambda_2)},
	\label{eq:BG_product_lambda}\\
	\lmat{C} &= \left(\lmat{C}_1^{-1}+\lmat{C}_2^{-1}\right)^{-1},
	\label{eq:BG_product_C}\\
	\gvec{\mu} &= \lmat{C}\left(\lmat{C}_1^{-1}\gvec{\mu}_1+\lmat{C}_2^{-1}\gvec{\mu}_2\right).
	\label{eq:BG_product_mu}
\end{align}
\end{corollary}
\vspace*{-2mm}
\begin{proof}
The proof is shown in Appendix~\ref{app:BG_product}.
\end{proof}
\vspace*{-1mm}
The mean~\eqref{eq:BG_product_mu} and covariance matrix~\eqref{eq:BG_product_C} of the Gaussian part are identical to those in the Gaussian product lemma\footnote{Gaussian product lemma~\cite{Bromiley2003},~\cite{Ngo2020}: $\CN{\lvec{x}|\gvec{\mu}_1,\lmat{C}_1} \cdot \CN{\lvec{x}|\gvec{\mu}_2,\lmat{C}_2} = \CN{\lvec{x}|\gvec{\mu},\lmat{C}} \cdot \CN{\lvec{0}|\gvec{\mu}_1-\gvec{\mu}_2,\lmat{C}_1+\lmat{C}_2}$ with $\lmat{C} = \big(\lmat{C}_1^{-1}+\lmat{C}_2^{-1}\big)^{-1}$ and $\gvec{\mu} = \lmat{C}\big(\lmat{C}_1^{-1}\gvec{\mu}_1+\lmat{C}_2^{-1}\gvec{\mu}_2\big).$}.
The activity probability~\eqref{eq:BG_product_lambda} provides insights when we explicitly consider the two events a \ac{BG} random variable characterizes.
The event of inactivity occurs with a probability that is proportional to the product of the two factor inactivity probabilities, i.e., $(1-\lambda_1)(1-\lambda_2)$.
The activity probability~\eqref{eq:BG_product_lambda} is proportional to the product of the two factor activity probabilities, i.e., $\lambda_1\lambda_2$, and a correction factor which measures how well the two Gaussian densities match, i.e., $\CN{\lvec{0}|\gvec{\mu}_1-\gvec{\mu}_2,\lmat{C}_1+\lmat{C}_2}$.
The correction factor is the normalization constant that appears in the Gaussian product lemma.
Note that for $\lambda_1=\lambda_2=1$, the \ac{BG} product lemma reduces to the Gaussian product lemma.

For illustration purposes, two examples of a \ac{BG} product are shown in Fig.~\ref{fig:BG_product}.
\begin{figure}[t]
    \centering
    \subfloat[$\mu_1 = -2$, $\mu_2 = 3$.]{\includegraphics[width=0.2\textwidth]{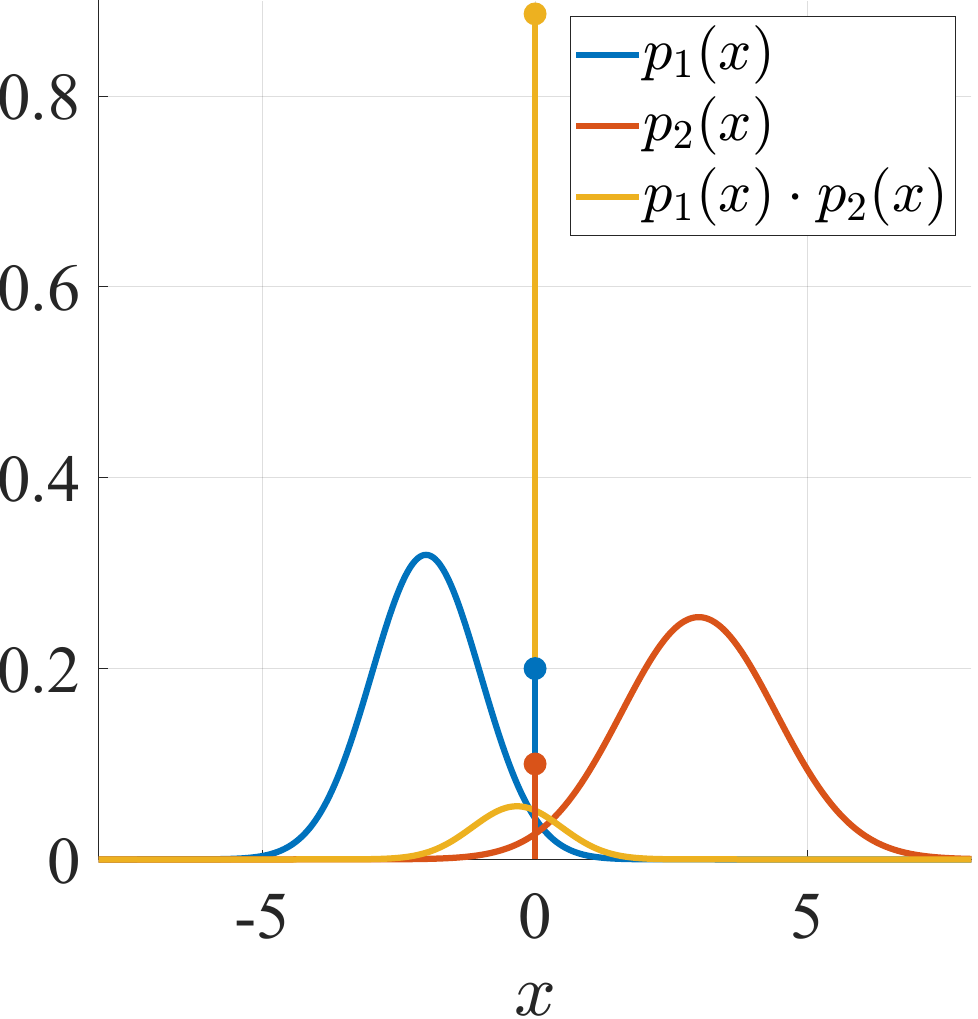}
        \label{fig:BG_product1}}\hfill
    \subfloat[$\mu_1 = -2$, $\mu_2 = -3$.]{\includegraphics[width=0.2\textwidth]{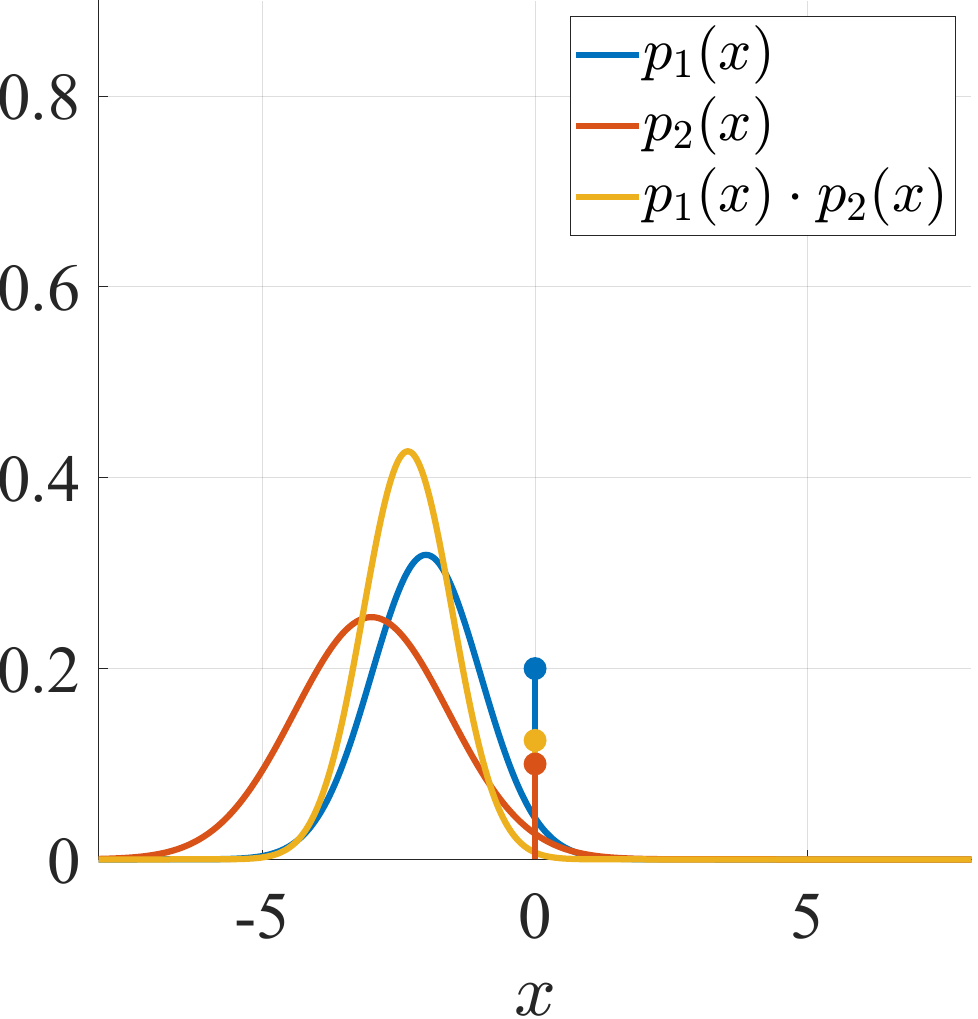}
        \label{fig:BG_product2}}
    \caption{Example BG product with $\lambda_1=0.8$ and $\lambda_2=0.9$.}
    \label{fig:BG_product}
    \vspace*{-5mm}
\end{figure}
Here, the random variable is modeled as real-valued for simplicity.
In Fig.~\ref{fig:BG_product1}, it can be observed that the area of the Gaussian component corresponding to the activity probability for the resulting product is fairly small, even though the two original \ac{BG} distributions have an activity probability of $\lambda_1=0.8$ and $\lambda_2=0.9$, respectively, because the respective Gaussian distributions do not significantly overlap.
The opposite behavior can be observed in Fig.~\ref{fig:BG_product2}.

Besides, it is worth mentioning that computing the product of a \ac{BG} and a Gaussian distribution is also straightforward with the proposed exponential family representation and reduces to simple addition of the natural parameters corresponding to the Gaussian component.
Assuming a \ac{BG} distribution with natural parameters $(\kappa_1,\gvec{\gamma}_1,\gmat{\Lambda}_1)$ and a Gaussian distribution with $(\gvec{\gamma}_2,\gmat{\Lambda}_2)$, the product is a \ac{BG} distribution with natural parameters $(\kappa_1,\gvec{\gamma}_1+\gvec{\gamma}_2,\gmat{\Lambda}_1+\gmat{\Lambda}_2)$.
The same result is obtained by modeling the \ac{BG} distribution as the sum of a weighted Dirac delta at zero and a weighted Gaussian distribution, i.e., $(1-\lambda_1)\delta(\lvec{x})+\lambda_1\CN{\lvec{x}|\gvec{\mu}_1,\lmat{C}_1}$, and then multiplying the sum by the Gaussian distribution $\CN{\lvec{x}|\gvec{\mu}_2,\lmat{C}_2}$ as proposed in~\cite{Vila2011}.
The equivalence of the two approaches is omitted for brevity.

Finally, we note that the proposed approach for the \ac{BG} exponential family representation can be readily extended to a broader class of probability distributions with an arbitrary number of discrete point masses at arbitrary positions.
Furthermore, the active part that is modeled by a continuous Gaussian distribution in our work can be replaced by another continuous distribution with exponential family representation.
The kind of distributions that can be modeled by our approach are related to the so-called \emph{spike and slab priors}~\cite{Lobato2013} where the discrete probability mass characterizes the spike and the continuous probability distribution models the slab.
This generality makes the presented approach applicable to a wide range of problems in signal processing.

\vspace*{-1mm}
\section{\acs{EP}-based \acs{JACD} Algorithms}\label{sec:JACD-EP-BG}
\vspace*{-1mm}
In this section, we propose two novel \ac{EP}-based \ac{JACD} algorithms.
The algorithms are derived by first introducing a convenient factorization of $p(\lmat{U},\lmat{H},\data{\lmat{X}}|\lmat{Y},\pilot{\lmat{X}})$ which induces a factor graph as shown in Section~\ref{subsec:FG}.
Then, we assign parametric exponential family representations to the factors to approximate the \ac{APP} distribution as discussed in Section \ref{subsec:EP_approx_fh_load}.
Different choices of these approximate exponential family models yield the two proposed algorithms.
Finally, \ac{EP} message-passing rules are applied to the factor graph as detailed in Section \ref{subsec:MP_updates}.

\vspace*{-2mm}
\subsection{Factor Graph Representation}\label{subsec:FG}
Similar to~\cite{Karataev2024}, we introduce the auxiliary variables $\lvec{g}_{l,k}\coloneq \lvec{h}_{l,k}u_k,$ and $\lvec{z}_{l,kt} \coloneq \lvec{g}_{l,k}x_{kt}$ to decouple activities, channels, and data across \acp{UE}.
We collect all auxiliary variables in the matrices $\lmat{G}$ and $\lmat{Z}$, respectively.
The joint \ac{MAP} estimator of the user activity, channel, data symbols, and auxiliary variables maximizes the \ac{APP} distribution given by
\vspace*{-2mm}
\begin{equation}
\begin{split}
	&p(\lmat{U},\lmat{H},\data{\lmat{X}},\lmat{G},\lmat{Z}|\lmat{Y},\pilot{\lmat{X}})\\
    &\!\quad\propto\!\prod_{l=1}^L\!\prod_{k=1}^K\!\prod_{t=1}^T\Big[p(\lvec{y}_{l,t}|\lvec{z}_{l,1t},\!...,\lvec{z}_{l,Kt})\!\cdot p(\lvec{z}_{l,kt}|\lvec{g}_{l,k},x_{kt})\\
	&\!\qquad\cdot p(\lvec{g}_{l,k}|\lvec{h}_{l,k},u_k)\!\cdot\tilde{p}_{u_k}(u_k)\!\cdot\tilde{p}_{h_{l,k}}(\lvec{h}_{l,k})\!\cdot p_x(x_{kt})\Big],
	\label{eq:APP_aux}
\end{split}
\end{equation}
where the factorization follows from the independence of channel vectors across \acp{AP} and \acp{UE}, the independence of user activities across \acp{UE}, and the independence of data symbols across \acp{UE} and time indices.
The terms $\tilde{p}_{u_k}(u_k)$ and $\tilde{p}_{h_{l,k}}(\lvec{h}_{l,k})$ denote refined prior information on the user activity $u_k$ and the channel $\lvec{h}_{l,k}$, respectively, which can be acquired by a pilot-based initialization algorithm.
One possible initialization algorithm is the \ac{JAC-EP} algorithm presented in~\cite{Forsch2024}.
We denote the resulting prior information for $u_k$ and $\lvec{h}_{l,k}$ by the two probabilities $\tilde{p}_{u_k}(0)$ and $\tilde{p}_{u_k}(1)$, and by the channel mean vector $\tilde{\gvec{\mu}}_{h_{l,k}}$ and covariance matrix $\tilde{\lmat{C}}_{h_{l,k}}$, respectively.
The probability distributions in~\eqref{eq:APP_aux} are represented by factor nodes (rectangles) in the factor graph illustrated in Fig.~\ref{fig:FG} and are given by
\vspace*{-3mm}
\begin{alignat*}{2}
    &\Psi_{y_{l,t}} &&\coloneq p(\lvec{y}_{l,t}|\lvec{z}_{l,1t},\!...,\lvec{z}_{l,Kt}) = \mathcal{CN}\bigg(\lvec{y}_{l,t}\bigg|\sum_{k=1}^K\lvec{z}_{l,kt},\sigma_n^2\lmat{I}_N\bigg),
    \\
    &\Psi_{z_{l,kt}} &&\coloneq p(\lvec{z}_{l,kt}|\lvec{g}_{l,k},x_{kt}) = \delta(\lvec{z}_{l,kt}-\lvec{g}_{l,k}x_{kt}),
    \\
    &\Psi_{g_{l,k}} &&\coloneq p(\lvec{g}_{l,k}|\lvec{h}_{l,k},u_k) = \delta(\lvec{g}_{l,k}-\lvec{h}_{l,k}u_k),
\end{alignat*}
\begin{alignat*}{2}
    &\Psi_{u_k} &&\coloneq \tilde{p}_{u_k}(u_k) = \tilde{p}_{u_k}(0)\ind{u_k=0}+\tilde{p}_{u_k}(1)\ind{u_k=1},
    \\
    &\Psi_{h_{l,k}} &&\coloneq \tilde{p}_{h_{l,k}}(\lvec{h}_{l,k}) = \mathcal{CN}\big(\lvec{h}_{l,k}\big|\tilde{\gvec{\mu}}_{h_{l,k}},\tilde{\lmat{C}}_{h_{l,k}}\big),
    \\
    &\Psi_{x_{kt}} &&\coloneq p_x(x_{kt}) = \begin{cases}
 {\ind{x_{kt}=\pilot{x}_{kt}}}&\text{for }t\leq T_p\\
 \frac{1}{M}\ind{x_{kt}\in\mathcal{X}}&\text{for }t>T_p
        \end{cases}.
\end{alignat*}
\begin{figure}[t]
    \centerline{\includegraphics[width=0.34\textwidth]{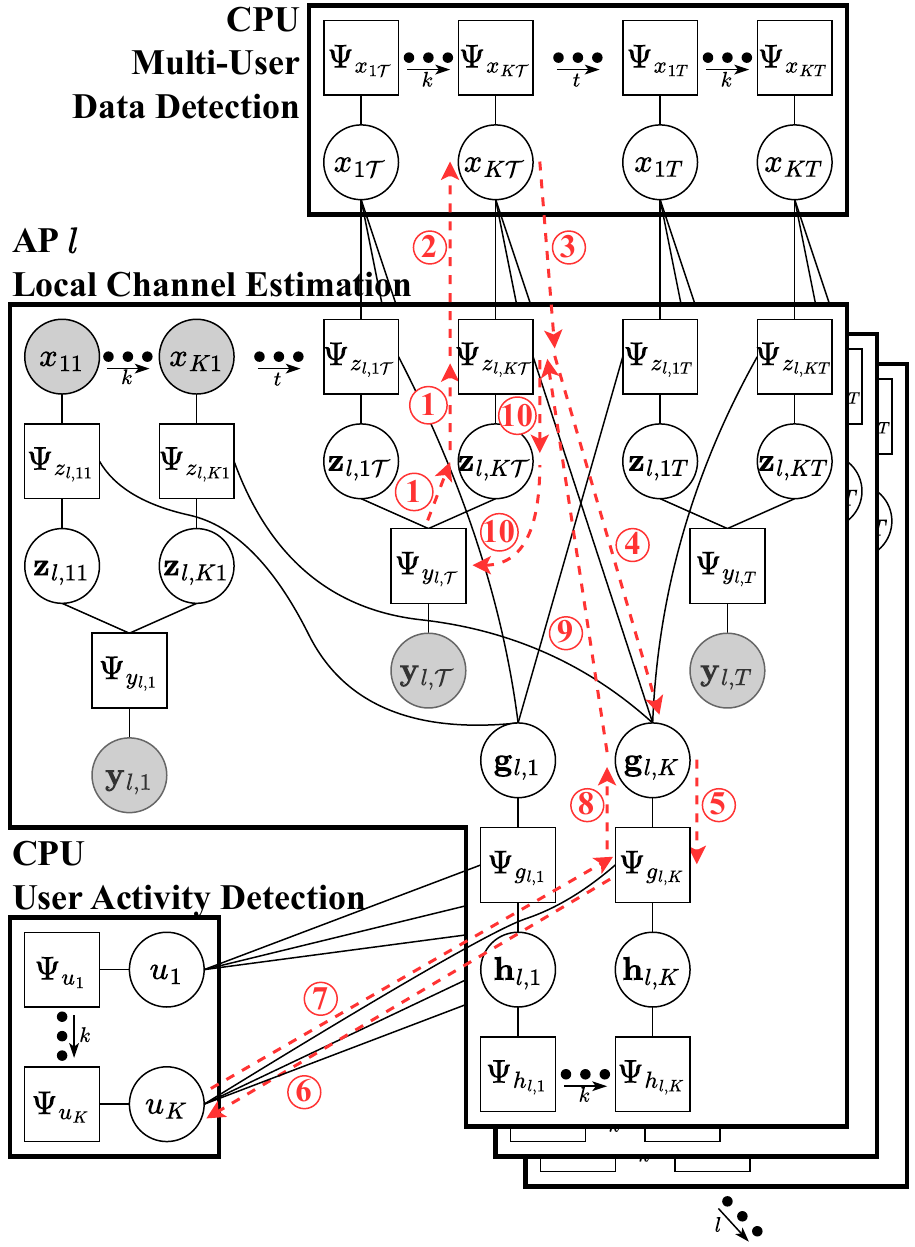}}
    \caption{Factor graph for the \ac{EP}-based \ac{JACD} algorithms with $\mathcal{T}\coloneq T_p+1$. The numbered red dashed arrows show the flow of information according to the scheduling presented in Algorithm~\ref{alg:JACD-EP-BG}. Each number corresponds to one message update in Algorithm~\ref{alg:JACD-EP-BG}.}
    \label{fig:FG}
    \vspace*{-2mm}
\end{figure}%
The variables in the above equations correspond to variable nodes (circles) in the factor graph.
The factor and variable nodes are organized to reflect their implementation at the \ac{CPU} and the \acp{AP}.
As shown in the factor graph, the \ac{CPU} combines the information from the \acp{AP} to estimate user activities and data.
In contrast, the channels are estimated locally in each \ac{AP} and do not need to be forwarded to the \ac{CPU}.

\vspace*{-2mm}
\subsection{\acs{EP} Approximations and Fronthaul Load}\label{subsec:EP_approx_fh_load}
To apply the \ac{EP} message-passing rules to the factor graph in Fig.~\ref{fig:FG}, we assign a parametric distribution representation from the exponential family to each variable node to approximate the corresponding  \ac{APP} distribution.
Categorical distributions are chosen for the variables $x_{kt}$ and $u_k$ whereas the variables $\lvec{z}_{l,kt}$ and $\lvec{h}_{l,k}$ are modeled by multivariate complex Gaussian distributions.
The variable $\lvec{g}_{l,k}$ is treated differently for the two proposed algorithms.
In the \ac{JACD-EP} algorithm, $\lvec{g}_{l,k}$ follows a multivariate complex Gaussian distribution, whereas in the \ac{JACD-EP-BG} algorithm, $\lvec{g}_{l,k}$ is modeled by a \ac{BG} vector as presented in Section~\ref{sec:BG}. The parameters characterizing these approximate posterior distributions constitute the messages propagated along the graph.

The fronthaul load is determined solely by the messages associated to $x_{kt}$ and $u_k.$ For $t>T_p, $ messages from and towards node $x_{kt}$  consist of $M-1$ probabilities while no messages are exchanged for $t\leq T_p$ since the pilot symbols are known a priori.
For $u_k$, a single real-valued parameter suffices to describe its distribution.
Hence, the total fronthaul load per iteration amounts to $2LK(T_d(M-1)+1)$ real-valued numbers where the factor two accounts for the bidirectional exchange between \ac{CPU} and \acp{AP}.

\vspace*{-2mm}
\subsection{Initialization and Scheduling}\label{subsec:init_sched_alg}
The proposed algorithms are initialized using the activity and channel estimates obtained from the \acs{JAC-EP}\acused{JAC-EP} algorithm~\cite{Forsch2024}.
This pilot-based initialization algorithm provides the soft user-activity estimates $\tilde{p}_{u_k}(u_k)$ and channel distributions $\tilde{p}_{h_{l,k}}(\lvec{h}_{l,k})$, equivalently characterized by the expectation of their sufficient statistics $(\tilde{\gvec{\mu}}_{h_{l,k}}, \tilde{\lmat{C}}_{h_{l,k}})$ or their natural parameters $(\tilde{\gvec{\gamma}}_{h_{l,k}}, \tilde{\gmat{\Lambda}}_{h_{l,k}})$, which serve as priors for the \ac{JACD-EP} and the \ac{JACD-EP-BG} algorithm.
The exchanged messages are initialized as follows.
The initial mean vector and covariance matrix for the message $\msg{\Psi_{z_{l,kt}}}{\lvec{z}_{l,kt}}$ $\forall k,l,t$ are determined from the prior information on $\lvec{z}_{l,kt}$ induced by $\tilde{p}_{u_k}(u_k)$, $\tilde{p}_{h_{l,k}}(\lvec{h}_{l,k})$, and $p_x(x_{kt})$,\footnote{We compute the expectation of the sufficient statistics of $\lvec{z}_{l,kt}$ by exploiting the independence of $u_k$, $\lvec{h}_{l,k}$, and $x_{kt}$, e.g., the mean $\est{}{\lvec{z}_{l,kt}}=\est{}{u_k}\est{}{\lvec{h}_{l,k}}\est{}{x_{kt}}$. The same approach is applied to  $\lvec{g}_{l,k}$.}
\begin{align}
	\Mumsg{\Psi_{z_{l,kt}}}{\lvec{z}_{l,kt}} \!\!&=\! \begin{cases}
 \tilde{p}_{u_k}(1)\cdot\tilde{\gvec{\mu}}_{h_{l,k}}\cdot\pilot{x}_{kt}&\text{for }t\leq T_p\\
 \lvec{0}&\text{for }t>T_p
        \end{cases},
	\label{eq:mu_Psi_z_z_init}\\
	\Cmsg{\Psi_{z_{l,kt}}}{\lvec{z}_{l,kt}} \!\!&=\! \begin{cases}
    \begin{aligned}
    \tilde{p}_{u_k}(1)\big(&\tilde{\lmat{C}}_{h_{l,k}}+\tilde{\gvec{\mu}}_{h_{l,k}}\tilde{\gvec{\mu}}_{h_{l,k}}\herm\\
    &\cdot\tilde{p}_{u_k}(0)\big)\;|\pilot{x}_{kt}|^2
    \end{aligned}
    &\text{for }t\leq T_p\\
 \tilde{p}_{u_k}(1)\big(\tilde{\lmat{C}}_{h_{l,k}}+\tilde{\gvec{\mu}}_{h_{l,k}}\tilde{\gvec{\mu}}_{h_{l,k}}\herm\big)\sigma_x^2&\text{for }t>T_p
        \end{cases}.
	\label{eq:C_Psi_z_z_init}
\end{align}
Similarly, the initialization of the messages $\msg{\Psi_{g_{l,k}}}{\lvec{g}_{l,k}}$ and $\msg{\lvec{g}_{l,k}}{\Psi_{z_{l,kt}}}$ $\forall k,l,t$ are determined from the prior information on $\lvec{g}_{l,k}$ and depends on the different \ac{EP} approximations for $\lvec{g}_{l,k}$.
The Gaussian variables $\lvec{g}_{l,k}$ in the \ac{JACD-EP} algorithm are initialized as
\begin{align}
	\Mumsg{\Psi_{g_{l,k}}}{\lvec{g}_{l,k}} \!\!&= \Mumsg{\lvec{g}_{l,k}}{\Psi_{z_{l,kt}}} \!\!= \tilde{p}_{u_k}(1)\cdot\tilde{\gvec{\mu}}_{h_{l,k}},
	\label{eq:mu_g_Psi_z_init_}\\
	\Cmsg{\Psi_{g_{l,k}}}{\lvec{g}_{l,k}} \!\!&= \Cmsg{\lvec{g}_{l,k}}{\Psi_{z_{l,kt}}} \!\!= \tilde{p}_{u_k}(1)\big(\tilde{\lmat{C}}_{h_{l,k}}\!\!+\!\tilde{\gvec{\mu}}_{h_{l,k}}\tilde{\gvec{\mu}}_{h_{l,k}}\herm\!\!\cdot\!\tilde{p}_{u_k}(0)\big).
	\label{eq:C_g_Psi_z_init_}
\end{align}
In the \ac{JACD-EP-BG} algorithm, the initialization of the \ac{BG} variables $\lvec{g}_{l,k}$ is given by
\begin{align}
	\Actmsg{\Psi_{g_{l,k}}}{\lvec{g}_{l,k}} &= \Actmsg{\lvec{g}_{l,k}}{\Psi_{z_{l,kt}}} = \tilde{p}_{u_k}(1),
	\label{eq:act_g_Psi_z_init}\\
	\Mumsg{\Psi_{g_{l,k}}}{\lvec{g}_{l,k}} &= \Mumsg{\lvec{g}_{l,k}}{\Psi_{z_{l,kt}}} = \tilde{\gvec{\mu}}_{h_{l,k}},
	\label{eq:mu_g_Psi_z_init}\\
	\Cmsg{\Psi_{g_{l,k}}}{\lvec{g}_{l,k}} &= \Cmsg{\lvec{g}_{l,k}}{\Psi_{z_{l,kt}}} = \tilde{\lmat{C}}_{h_{l,k}}.
	\label{eq:C_g_Psi_z_init}
\end{align}
All other messages are initialized using uninformative priors, namely, uniform distributions for messages involving categorical variables, and zero-mean, zero-precision distributions for messages involving Gaussian variables where the precision matrix is the inverse of the covariance matrix.
The uninformative \ac{BG} distribution exhibits an activity probability equal to 0.5 and a zero-mean, zero-precision Gaussian component.

After initialization, messages are updated according to the scheduling defined in Algorithm~\ref{alg:JACD-EP-BG} and illustrated in Fig.~\ref{fig:FG} by the red dashed arrows.
\begin{algorithm}[t]
\caption{\acs{JACD-EP}$^*$ and \acs{JACD-EP-BG}$^\dagger$ Algorithm}
\begin{algorithmic}[1]
\renewcommand{\algorithmicrequire}{\textbf{Input:}}
\renewcommand{\algorithmicensure}{\textbf{Output:}}
\REQUIRE Pilot matrix $\pilot{\lmat{X}}$, transmit power $\sigma_x^2$, received signal $\lmat{Y}$, noise variance $\sigma_n^2$, prior distributions on user activities $\tilde{p}_{u_k}(u_k)$ and channels $\tilde{p}_{h_{l,k}}(\lvec{h}_{l,k})$.
\ENSURE Estimated activities $\hat{u}_k$, channels $\hat{\lmat{h}}_{l,k}$, and data $\hat{x}_{kt}$.\hspace{-1mm}
\STATE $\forall k,l,t$: Initialize $\msg{\Psi_{z_{l,kt}}}{\lvec{z}_{l,kt}}$ \eqref{eq:mu_Psi_z_z_init} \eqref{eq:C_Psi_z_z_init}.
\STATE $\forall k,l,t$: Initialize $\msg{\lvec{g}_{l,k}}{\Psi_{z_{l,kt}}}$ \{\eqref{eq:mu_g_Psi_z_init_} \eqref{eq:C_g_Psi_z_init_}\}$^*$ \{\eqref{eq:act_g_Psi_z_init} - \eqref{eq:C_g_Psi_z_init}\}$^\dagger$.\hspace{-1mm}
\STATE $\forall k,l$: Initialize $\msg{\Psi_{g_{l,k}}}{\lvec{g}_{l,k}}$ \{\eqref{eq:mu_g_Psi_z_init_} \eqref{eq:C_g_Psi_z_init_}\}$^*$ \{\eqref{eq:act_g_Psi_z_init} - \eqref{eq:C_g_Psi_z_init}\}$^\dagger$.
\FOR {$i = 1$ to $i_\text{max}$}
\STATE $\forall k,l,t$: Update $\msg{\Psi_{y_{l,t}}}{\lvec{z}_{l,kt}}$ \eqref{eq:mu_Psi_y_z} \eqref{eq:C_Psi_y_z}.
\STATE $\forall k,l,t\!>\!T_p$: Update $\msg{\Psi_{z_{l,kt}}}{x_{kt}}$ \eqref{eq:m_Psi_z_x}.\label{alg_line:m_Psi_z_x}
\STATE $\forall k,l,t\!>\!T_p$: Update $\msg{x_{kt}}{\Psi_{z_{l,kt}}}$ \eqref{eq:m_x_Psi_z}.
\STATE $\forall k,l,t$: Update $\msg{\Psi_{z_{l,kt}}}{\lvec{g}_{l,k}}$ \{\eqref{eq:mu_Psi_z_g} \eqref{eq:C_Psi_z_g} via \eqref{eq:mu_mmd_Psi_z_g_} \eqref{eq:C_mmd_Psi_z_g_}\}$^*$ \{\eqref{eq:kappa_Psi_z_g} - \eqref{eq:C_Psi_z_g} via \eqref{eq:act_mmd_Psi_z_g} - \eqref{eq:Lambda_mmd_Psi_z_g}\}$^\dagger$.\label{alg_line:m_Psi_z_g}\hspace{-1mm}
\STATE $\forall k,l$: Update $\msg{\lvec{g}_{l,k}}{\Psi_{g_{l,k}}}$ \{\eqref{eq:mu_g_Psi_g} \eqref{eq:C_g_Psi_g}\}$^*$ \{\eqref{eq:kappa_g_Psi_g} - \eqref{eq:C_g_Psi_g}\}$^\dagger$.
\STATE $\forall k,l$: Update $\msg{\Psi_{g_{l,k}}}{u_k}$ \eqref{eq:m_Psi_g_u_}$^*$ \eqref{eq:m_Psi_g_u}$^\dagger$.\label{alg_line:m_Psi_g_u}
\STATE $\forall k,l$: Update $\msg{u_k}{\Psi_{g_{l,k}}}$ \eqref{eq:m_u_Psi_g}.
\STATE $\forall k,l$: Update $\msg{\Psi_{g_{l,k}}}{\lvec{g}_{l,k}}$ \{\eqref{eq:mu_Psi_g_g_} \eqref{eq:C_Psi_g_g_}\}$^*$ \{\eqref{eq:act_Psi_g_g} - \eqref{eq:C_Psi_g_g}\}$^\dagger$.\label{alg_line:m_Psi_g_g}
\STATE $\forall k,l,t$: Update $\msg{\lvec{g}_{l,k}}{\Psi_{z_{l,kt}}}$ \{\eqref{eq:mu_g_Psi_z} \eqref{eq:C_g_Psi_z}\}$^*$ \{\eqref{eq:kappa_g_Psi_z} - \eqref{eq:C_g_Psi_z}\}$^\dagger$.\hspace{-1mm}
\STATE $\forall k,l,t$: Update $\msg{\Psi_{z_{l,kt}}}{\lvec{z}_{l,kt}}$ \eqref{eq:mu_Psi_z_z} \eqref{eq:C_Psi_z_z}.\label{alg_line:m_Psi_z_z}
\ENDFOR
\RETURN $\hat{u}_k$ \eqref{eq:estimate_u} $\forall k$.
\RETURN $\hat{\lvec{h}}_{l,k}$ \eqref{eq:estimate_h_explicit} $\forall k,l$.
\RETURN $\hat{x}_{kt}$ \eqref{eq:estimate_x} $\forall k,{t>T_p}$.
\end{algorithmic} 
\label{alg:JACD-EP-BG}
\end{algorithm}
Algorithm~\ref{alg:JACD-EP-BG} jointly describes both proposed methods which follow the same scheduling and differ only in the message updates.
We group update rules for the two algorithms by brackets with different superscripts, i.e., $\{\cdot\}^*$ for \ac{JACD-EP} and $\{\cdot\}^\dagger$ for \ac{JACD-EP-BG}.
Message updates without brackets are common to both algorithms.

\vspace*{-2mm}
\subsection{Message-Passing Update Rules}\label{subsec:MP_updates}
In this section, we present the \ac{EP} message-passing rules used in the proposed algorithms.
Detailed derivations for the \ac{JACD-EP} and the \ac{JACD-EP-BG} algorithm can be found in the extended version of~\cite{Forsch2024}\footnote{https://arxiv.org/abs/2405.09914} and Appendix~\ref{app:MP_updates}, respectively.
For a \ac{BG} random variable, the activity probability $\lambda$, the mean vector $\gvec{\mu}$, and the covariance matrix $\lmat{C}$ can be readily expressed by the natural parameters $\kappa=\log\frac{1-\lambda}{\lambda}+\gvec{\mu}\herm\lmat{C}^{-1}\gvec{\mu}+\log|\pi\lmat{C}|$, $\gvec{\gamma}\!=\!\lmat{C}^{-1}\gvec{\mu}$, and $\gmat{\Lambda}\!=\!\lmat{C}^{-1}$.
The same relations hold for the parameters $\gvec{\mu}$ and $\lmat{C}$ and natural parameters $\gvec{\gamma}$ and $\gmat{\Lambda}$ of a Gaussian random variable.
In the following, we freely switch between these two representations without explicitly mentioning the transformation; whenever $(\Actmsg{\Psi_\alpha}{\lvec{x}_\beta},\Mumsg{\Psi_\alpha}{\lvec{x}_\beta},\Cmsg{\Psi_\alpha}{\lvec{x}_\beta})$ are computed, the corresponding $(\Kappamsg{\Psi_\alpha}{\lvec{x}_\beta},\Gammamsg{\Psi_\alpha}{\lvec{x}_\beta},\Lambdamsg{\Psi_\alpha}{\lvec{x}_\beta})$ follow automatically and vice versa.

In the first step, the messages $\msg{\Psi_{y_{l,t}}}{\lvec{z}_{l,kt}}$ $\forall k,l,t$ are updated.
These updates compute the Gaussian beliefs of the auxiliary variables $\lvec{z}_{l,kt}$ based on the observation $\lvec{y}_{l,t}$.
In practice, the updating rules perform soft interference cancellation using the current estimated interference $\Mumsg{\Psi_{z_{l,k't}}}{\lvec{z}_{l,k't}}$ of all other users $k'\neq k$ and modifying the covariance matrix accordingly.
The message mean and covariance are
\begin{align}
    \Mumsg{\Psi_{y_{l,t}}}{\lvec{z}_{l,kt}} &= \lvec{y}_{l,t}-\sum_{k'\neq k}\Mumsg{\Psi_{z_{l,k't}}}{\lvec{z}_{l,k't}},
    \label{eq:mu_Psi_y_z}\\
    \Cmsg{\Psi_{y_{l,t}}}{\lvec{z}_{l,kt}} &= \sigma_n^2\lmat{I}_N+\sum_{k'\neq k}\Cmsg{\Psi_{z_{l,k't}}}{\lvec{z}_{l,k't}}.
    \label{eq:C_Psi_y_z}
\end{align}

These updated beliefs are used to generate the local categorical data symbol beliefs $\msg{\Psi_{z_{l,kt}}}{x_{kt}}$ $\forall k,l,t>T_p$ at each \ac{AP}.
The message update is derived by evaluating how well the product $\lvec{g}_{l,k}x_{kt}$ matches the relation $\lvec{z}_{l,kt}=\lvec{g}_{l,k}x_{kt}$ which is obtained by sampling a Gaussian likelihood $\theta(x_{kt})$,\footnote{Note that $\theta(x_{kt})$ provides unnormalized probability values and $\catmsg{\Psi_{z_{l,kt}}}{x_{kt}}(x_{kt})$ is obtained by normalization.}
\begin{equation}
    \catmsg{\Psi_{z_{l,kt}}}{x_{kt}}(x_{kt}) \propto \theta(x_{kt}),
    \label{eq:m_Psi_z_x}
\end{equation}
with
\begin{equation}
\begin{split}
    \theta(x_{kt}) = \mathcal{CN}\big(\lvec{0}|&\Mumsg{\Psi_{y_{l,t}}}{\lvec{z}_{l,kt}}-\Mumsg{\lvec{g}_{l,k}}{\Psi_{z_{l,kt}}}x_{kt},\\
 &\Cmsg{\Psi_{y_{l,t}}}{\lvec{z}_{l,kt}}+\Cmsg{\lvec{g}_{l,k}}{\Psi_{z_{l,kt}}}|x_{kt}|^2\big).
\end{split}
    \label{eq:theta_mmd_Psi_z_z}
\end{equation}

Then, all local data symbol beliefs are combined at the \ac{CPU}.
The message update $\msg{x_{kt}}{\Psi_{z_{l,kt}}}$ $\forall k,l,$ $ t>T_p$ combines the categorical beliefs of the data symbol $x_{kt}$ from all \acp{AP} $l'\neq l$ and forwards them to \ac{AP} $l$,
\begin{equation}
    \catmsg{x_{kt}}{\Psi_{z_{l,kt}}}(x_{kt}) \propto \prod_{l'\neq l}\catmsg{\Psi_{z_{l',kt}}}{x_{kt}}(x_{kt}).
    \label{eq:m_x_Psi_z}
\end{equation}

The updated beliefs of the variables $\lvec{z}_{l,kt}$ and $x_{kt}$ are used to update the messages $\msg{\Psi_{z_{l,kt}}}{\lvec{g}_{l,k}}$ $\forall k,l,t$, yielding the beliefs of the auxiliary variable $\lvec{g}_{l,k}$.
The \ac{EP} approach first generates the local estimate of $\lvec{g}_{l,k}$ at the factor node $\Psi_{z_{l,kt}}$ and then removes the knowledge induced by the message $\msg{\lvec{g}_{l,k}}{\Psi_{z_{l,kt}}}$, 
\vspace*{-1mm}
\begin{align}
    \Kappamsg{\Psi_{z_{l,kt}}}{\lvec{g}_{l,k}} &= \Kappamsgb{\lvec{g}_{l,kt}}{}-\Kappamsg{\lvec{g}_{l,k}}{\Psi_{z_{l,kt}}},
    \label{eq:kappa_Psi_z_g}\\
    \Gammamsg{\Psi_{z_{l,kt}}}{\lvec{g}_{l,k}} &= \Gammamsgb{\lvec{g}_{l,kt}}{}-\Gammamsg{\lvec{g}_{l,k}}{\Psi_{z_{l,kt}}},
    \label{eq:mu_Psi_z_g}\\
    \Lambdamsg{\Psi_{z_{l,kt}}}{\lvec{g}_{l,k}} &= \Lambdamsgb{\lvec{g}_{l,kt}}{}-\Lambdamsg{\lvec{g}_{l,k}}{\Psi_{z_{l,kt}}}.
    \label{eq:C_Psi_z_g}
\end{align}
Note that~\eqref{eq:kappa_Psi_z_g} is computed only for the \ac{JACD-EP-BG} algorithm, whereas~\eqref{eq:mu_Psi_z_g} and~\eqref{eq:C_Psi_z_g} are evaluated for both algorithms.
The computation of the parameters describing the local estimate of $\lvec{g}_{l,k}$ at $\Psi_{z_{l,kt}}$, i.e., $\Kappamsgb{\lvec{g}_{l,kt}}{}$, $\Gammamsgb{\lvec{g}_{l,kt}}{}$, and $\Lambdamsgb{\lvec{g}_{l,kt}}{}$, differs for the \ac{JACD-EP} and the \ac{JACD-EP-BG} algorithm due to the different modeling of $\lvec{g}_{l,k}$.
For the \ac{JACD-EP} algorithm, we compute the Gaussian parameters
\vspace*{-1mm}
\begin{align}
    \Mumsgb{\lvec{g}_{l,kt}}{} &= \frac{1}{{Z}_{\Psi_{z_{l,kt}}}}\sum_{x_{kt}\in\tilde{\mathcal{X}}}\frac{\phi(x_{kt})}{x_{kt}}\cdot\Mumsga{\lvec{z}_{l,kt}}{}(x_{kt}),
	\label{eq:mu_mmd_Psi_z_g_}\\
    \begin{split}
	\Cmsgb{\lvec{g}_{l,kt}}{} &= \frac{1}{{Z}_{\Psi_{z_{l,kt}}}}\sum_{x_{kt}\in\tilde{\mathcal{X}}}\frac{\phi(x_{kt})}{|x_{kt}|^2}\cdot\big(\Cmsga{\lvec{z}_{l,kt}}{}(x_{kt})\\
    &\qquad+\Mumsga{\lvec{z}_{l,kt}}{}(x_{kt})\cdot\Mumsga{\lvec{z}_{l,kt}}{}\herm(x_{kt})\big) - \Mumsgb{\lvec{g}_{l,kt}}{}\Mumsgb{\lvec{g}_{l,kt}}{}\herm,
    \end{split}
    \label{eq:C_mmd_Psi_z_g_}
\end{align}
with $\tilde{\mathcal{X}}=\{\pilot{x}_{kt}\}$ for $t\leq T_p$, $\tilde{\mathcal{X}}=\mathcal{X}$ for $t>T_p$, and
\begin{align}
    \phi(x_{kt}) &= \catmsg{x_{kt}}{\Psi_{z_{l,kt}}}(x_{kt})\cdot\theta(x_{kt}),
    \label{eq:phi_mmd_Psi_z}\\
	\Gammamsga{\lvec{z}_{l,kt}}{}(x_{kt}) &= \Gammamsg{\Psi_{y_{l,t}}}{\lvec{z}_{l,kt}}+\Gammamsg{\lvec{g}_{l,k}}{\Psi_{z_{l,kt}}}\frac{x_{kt}}{|x_{kt}|^2},
	\label{eq:mu_tmp_mmd_Psi_z_z}\\
    \Lambdamsga{\lvec{z}_{l,kt}}{}(x_{kt}) &= \Lambdamsg{\Psi_{y_{l,t}}}{\lvec{z}_{l,kt}}+\Lambdamsg{\lvec{g}_{l,k}}{\Psi_{z_{l,kt}}}|x_{kt}|^{-2},
	\label{eq:C_tmp_mmd_Psi_z_z}\\
    {Z}_{\Psi_{z_{l,kt}}} &= \sum_{x_{kt}\in\tilde{\mathcal{X}}}\phi(x_{kt}).
	\label{eq:Z_Psi_z}
\end{align}
Note that $\catmsg{x_{kt}}{\Psi_{z_{l,kt}}}(\pilot{x}_{kt})=1$ for $t\leq T_p$ since the pilot symbol is known a priori.
Hence, the updates for the pilot part simplify to $\Mumsgb{\lvec{g}_{l,kt}}{}=\Mumsga{\lvec{z}_{l,kt}}{}(\pilot{x}_{kt})/\pilot{x}_{kt}$ and $\Cmsgb{\lvec{g}_{l,kt}}{}=\Cmsga{\lvec{z}_{l,kt}}{}(\pilot{x}_{kt})/|\pilot{x}_{kt}|^2$.
For the \ac{JACD-EP-BG} algorithm, we compute the parameters of the local \ac{BG} estimate of $\lvec{g}_{l,k}$ as
\vspace*{-2mm}
\begin{align}
    \Actmsgb{\lvec{g}_{l,kt}}{} &= \frac{\Actmsg{\lvec{g}_{l,k}}{\Psi_{z_{l,kt}}}\cdot\phi(x^*)}{\Actmsg{\lvec{g}_{l,k}}{\Psi_{z_{l,kt}}}\cdot\phi(x^*)+(1-\Actmsg{\lvec{g}_{l,k}}{\Psi_{z_{l,kt}}})\cdot\theta(0)},
    \label{eq:act_mmd_Psi_z_g}\\
    \Gammamsgb{\lvec{g}_{l,kt}}{} &= \Gammamsg{\Psi_{y_{l,t}}}{\lvec{z}_{l,kt}}\frac{|x^*|^2}{x^*}+\Gammamsg{\lvec{g}_{l,k}}{\Psi_{z_{l,kt}}},
    \label{eq:gamma_mmd_Psi_z_g}\\
    \Lambdamsgb{\lvec{g}_{l,kt}}{} &= \Lambdamsg{\Psi_{y_{l,t}}}{\lvec{z}_{l,kt}}|x^*|^2+\Lambdamsg{\lvec{g}_{l,k}}{\Psi_{z_{l,kt}}},
    \label{eq:Lambda_mmd_Psi_z_g}
\end{align}
with $x^*=\pilot{x}_{kt}$ for $t\leq T_p$ and
\vspace*{-1mm}
\begin{equation}
    x^* = \argmax_{x_{kt}\in\mathcal{X}}\phi(x_{kt}),
    \label{eq:x_star_mmd_Psi_z_g}
\end{equation}
for $t>T_p$.
Note that using~\eqref{eq:gamma_mmd_Psi_z_g} and~\eqref{eq:Lambda_mmd_Psi_z_g} in~\eqref{eq:mu_Psi_z_g} and~\eqref{eq:C_Psi_z_g}, the message updates simplify to $\Gammamsg{\Psi_{z_{l,kt}}}{\lvec{g}_{l,k}} = \Gammamsg{\Psi_{y_{l,t}}}{\lvec{z}_{l,kt}}|x^*|^2/x^*$ and $\Lambdamsg{\Psi_{z_{l,kt}}}{\lvec{g}_{l,k}} = \Lambdamsg{\Psi_{y_{l,t}}}{\lvec{z}_{l,kt}}|x^*|^2$, respectively.
Furthermore, note that we deviate here from the classical \ac{EP} message-passing rule by computing the local estimate of $\lvec{g}_{l,k}$ using the most likely symbol $x^*$ instead of averaging across all $x_{kt}\in\mathcal{X}$.
This is necessary in order to prevent the Gaussian part of the resulting \ac{BG} belief to be close to zero which can lead to a high number of false alarms.

Next, the messages $\msg{\lvec{g}_{l,k}}{\Psi_{g_{l,k}}}$ $\forall k,l$ are updated by combining the beliefs of $\lvec{g}_{l,k}$ from all time slots $t\in\{1,\dots,T\}$ to generate the updated belief parameters
\vspace*{-1mm}
\begin{align}
    \Kappamsg{\lvec{g}_{l,k}}{\Psi_{g_{l,k}}} &= \sum_{t=1}^T\Kappamsg{\Psi_{z_{l,kt}}}{\lvec{g}_{l,k}},
    \label{eq:kappa_g_Psi_g}\\
    \Gammamsg{\lvec{g}_{l,k}}{\Psi_{g_{l,k}}} &= \sum_{t=1}^T\Gammamsg{\Psi_{z_{l,kt}}}{\lvec{g}_{l,k}},
    \label{eq:mu_g_Psi_g}\\
    \Lambdamsg{\lvec{g}_{l,k}}{\Psi_{g_{l,k}}} &= \sum_{t=1}^T\Lambdamsg{\Psi_{z_{l,kt}}}{\lvec{g}_{l,k}}.
    \label{eq:C_g_Psi_g}
\end{align}
Similar as for the message $\msg{\Psi_{z_{l,kt}}}{\lvec{g}_{l,k}}$,~\eqref{eq:kappa_g_Psi_g} is computed only for the \ac{JACD-EP-BG} algorithm, whereas~\eqref{eq:mu_g_Psi_g} and~\eqref{eq:C_g_Psi_g} are evaluated for both algorithms.

The updated beliefs of $\lvec{g}_{l,k}$ are then used to compute the local activity probabilities $\msg{\Psi_{g_{l,k}}}{u_k}$ $\forall k,l$ at each \ac{AP}.
Here, the update differs for the two proposed algorithms.
The categorical belief for the \ac{JACD-EP} algorithm is computed by evaluating how well the Gaussian belief of $\lvec{g}_{l,k}$ matches with the prior of the channel $\lvec{h}_{l,k}$,
\vspace*{-1mm}
\begin{equation}
	\catmsg{\Psi_{g_{l,k}}}{u_k}(u_k) \propto \vartheta(u_k),
	\label{eq:m_Psi_g_u_}\vspace*{-1mm}
\end{equation}
with
\begin{equation}
	\vartheta(u_k) = \mathcal{CN}\big(\lvec{0}|\Mumsg{\lvec{g}_{l,k}}{\Psi_{g_{l,k}}}\!\!-\tilde{\gvec{\mu}}_{h_{l,k}}u_k,\Cmsg{\lvec{g}_{l,k}}{\Psi_{g_{l,k}}}\!\!+\tilde{\lmat{C}}_{h_{l,k}}u_k\big).
	\label{eq:theta_mmd_Psi_g_u}
\end{equation}
For the \ac{JACD-EP-BG} algorithm, the activity belief incorporates both the Bernoulli and the Gaussian component of the \ac{BG} variable $\lvec{g}_{l,k}$,
\vspace*{-2mm}
\begin{equation}
	\catmsg{\Psi_{g_{l,k}}}{u_k}(u_k) \propto \begin{cases}
	1-\Actmsg{\lvec{g}_{l,k}}{\Psi_{g_{l,k}}} & \text{for }u_k=0\\
	\Actmsg{\lvec{g}_{l,k}}{\Psi_{g_{l,k}}}\cdot\vartheta(1) & \text{for }u_k=1
	\end{cases}.
	\label{eq:m_Psi_g_u}
\end{equation}

Then, for both the proposed algorithms, the \ac{CPU} computes the message updates $\msg{u_k}{\Psi_{g_{l,k}}}$ $\forall k,l$ by combining the categorical beliefs of $u_k$ from all \acp{AP} $l'\neq l$ with the prior information and forwards them to \ac{AP} $l$,
\begin{equation}
    \catmsg{u_k}{\Psi_{g_{l,k}}}(u_k) \propto \tilde{p}_{u_k}(u_k)\cdot\prod_{l'\neq l}\catmsg{\Psi_{g_{l',k}}}{u_{k}}(u_k).
    \label{eq:m_u_Psi_g}
\end{equation}

The belief of $\lvec{g}_{l,k}$ is then updated in the message $\msg{\Psi_{g_{l,k}}}{\lvec{g}_{l,k}}$ $\forall k,l$ by combining the information from the Bernoulli variable $u_k$ and the Gaussian variable $\lvec{h}_{l,k}$.
For the \ac{JACD-EP} algorithm, this is achieved by computing the local Gaussian estimate of $\lvec{g}_{l,k}$ at the factor node $\Psi_{g_{l,k}}$ and then removing the contribution of $\msg{\lvec{g}_{l,k}}{\Psi_{g_{l,k}}}$,
\begin{align}
    \Gammamsg{\Psi_{g_{l,k}}}{\lvec{g}_{l,k}} &= \Gammamsgb{\lvec{g}_{l,k}}{}-\Gammamsg{\lvec{g}_{l,k}}{\Psi_{g_{l,k}}},
	\label{eq:mu_Psi_g_g_}\\
	\Lambdamsg{\Psi_{g_{l,k}}}{\lvec{g}_{l,k}} &= \Lambdamsgb{\lvec{g}_{l,k}}{}-\Lambdamsg{\lvec{g}_{l,k}}{\Psi_{g_{l,k}}},
	\label{eq:C_Psi_g_g_}
\end{align}
with
\begin{align}
	\Mumsgb{\lvec{g}_{l,k}}{} &= \frac{1}{{Z}_{\Psi_{g_{l,k}}}}\!\cdot\!\catmsg{u_k}{\Psi_{g_{l,k}}}\!(1)\cdot\vartheta(1)\cdot\Mumsga{\lvec{g}_{l,k}}{},\label{eq:mu_mmd_Psi_g_g_}
\end{align}
\begin{align}
    \begin{split}
	\Cmsgb{\lvec{g}_{l,k}}{} &= \frac{1}{{Z}_{\Psi_{g_{l,k}}}}\!\cdot\!\catmsg{u_k}{\Psi_{g_{l,k}}}\!(1)\cdot\vartheta(1)\cdot(\Cmsga{\lvec{g}_{l,k}}{}\!+\!\Mumsga{\lvec{g}_{l,k}}{}\Mumsga{\lvec{g}_{l,k}}{}\herm)\!\!\!\!\\
    &\qquad- \Mumsgb{\lvec{g}_{l,k}}{}\Mumsgb{\lvec{g}_{l,k}}{}\herm,
    \end{split}
    \label{eq:C_mmd_Psi_g_g_}
\end{align}
and
\vspace*{-2mm}
\begin{align}
    \Gammamsga{\lvec{g}_{l,k}}{} &= \Gammamsg{\lvec{g}_{l,k}}{\Psi_{g_{l,k}}}+\tilde{\gvec{\gamma}}_{h_{l,k}},\label{eq:mu_tmp_mmd_Psi_g_h_}\\
    \Lambdamsga{\lvec{g}_{l,k}}{} &= \Lambdamsg{\lvec{g}_{l,k}}{\Psi_{g_{l,k}}}+\tilde{\gmat{\Lambda}}_{h_{l,k}},\label{eq:C_tmp_mmd_Psi_g_h_}\\
    {Z}_{\Psi_{g_{l,k}}} &= \catmsg{u_k}{\Psi_{g_{l,k}}}(0)\cdot\vartheta(0)+\catmsg{u_k}{\Psi_{g_{l,k}}}(1)\cdot\vartheta(1).\label{eq:Z_Psi_g}
\end{align}
For the \ac{JACD-EP-BG} algorithm, the \ac{BG} belief of $\lvec{g}_{l,k}$ is given by
\vspace*{-2mm}
\begin{align}
	\Actmsg{\Psi_{g_{l,k}}}{\lvec{g}_{l,k}} &= \catmsg{u_k}{\Psi_{g_{l,k}}}(1),
	\label{eq:act_Psi_g_g}\\
    \Mumsg{\Psi_{g_{l,k}}}{\lvec{g}_{l,k}} &= \tilde{\gvec{\mu}}_{h_{l,k}},
	\label{eq:mu_Psi_g_g}\\
	\Cmsg{\Psi_{g_{l,k}}}{\lvec{g}_{l,k}} &= \tilde{\lmat{C}}_{h_{l,k}}.
	\label{eq:C_Psi_g_g}
\end{align}

Next, the messages $\msg{\lvec{g}_{l,k}}{\Psi_{z_{l,kt}}}$ $\forall k,l,t$ are updated by combining the previously updated belief of $\lvec{g}_{l,k}$ with the contributions from all time slots $t'\neq t$,
\begin{align}
    \Kappamsg{\lvec{g}_{l,k}}{\Psi_{z_{l,kt}}} &= \Kappamsg{\Psi_{g_{l,k}}}{\lvec{g}_{l,k}}+\sum_{t'\neq t}\Kappamsg{\Psi_{z_{l,kt'}}}{\lvec{g}_{l,k}},
    \label{eq:kappa_g_Psi_z}\\
    \Gammamsg{\lvec{g}_{l,k}}{\Psi_{z_{l,kt}}} &= \Gammamsg{\Psi_{g_{l,k}}}{\lvec{g}_{l,k}}+\sum_{t'\neq t}\Gammamsg{\Psi_{z_{l,kt'}}}{\lvec{g}_{l,k}},
    \label{eq:mu_g_Psi_z}\\
    \Lambdamsg{\lvec{g}_{l,k}}{\Psi_{z_{l,kt}}} &= \Lambdamsg{\Psi_{g_{l,k}}}{\lvec{g}_{l,k}}+\sum_{t'\neq t}\Lambdamsg{\Psi_{z_{l,kt'}}}{\lvec{g}_{l,k}}.
    \label{eq:C_g_Psi_z}
\end{align}
The parameter in~\eqref{eq:kappa_g_Psi_z} is computed only for the \ac{JACD-EP-BG} algorithm, whereas~\eqref{eq:mu_g_Psi_z} and~\eqref{eq:C_g_Psi_z} are common to both algorithms.

Lastly, the messages $\msg{\Psi_{z_{l,kt}}}{\lvec{z}_{l,kt}}$ $\forall k,l,t$ are updated to form the Gaussian belief of $\lvec{z}_{l,kt}$ which will be used in the next iteration for interference cancellation.
The factor node $\Psi_{z_{l,kt}}$ first computes the local estimate of $\lvec{z}_{l,kt}$ based on the beliefs of $\lvec{g}_{l,k}$ and $x_{kt}$ and then removes the contribution of $\msg{\lvec{z}_{l,kt}}{\Psi_{z_{l,kt}}}$,
\begin{align}
    \Gammamsg{\Psi_{z_{l,kt}}}{\lvec{z}_{l,kt}} &= \Gammamsgb{\lvec{z}_{l,kt}}{}-\Gammamsg{\Psi_{y_{l,t}}}{\lvec{z}_{l,kt}},
    \label{eq:mu_Psi_z_z}\\
    \Lambdamsg{\Psi_{z_{l,kt}}}{\lvec{z}_{l,kt}} &= \Lambdamsgb{\lvec{z}_{l,kt}}{}-\Lambdamsg{\Psi_{y_{l,t}}}{\lvec{z}_{l,kt}},
    \label{eq:C_Psi_z_z}
\end{align}
with the parameters describing the local estimate of $\lvec{z}_{l,kt}$,
\begin{align}
    \Mumsgb{\lvec{z}_{l,kt}}{} &= \frac{1}{{Z}_{{l,kt}}}\sum_{x_{kt}\in\tilde{\mathcal{X}}}\phi(x_{kt})\cdot\Mumsga{\lvec{z}_{l,kt}}{}(x_{kt}),
    \label{eq:mu_mmd_Psi_z_z}\\
    \begin{split}
    \Cmsgb{\lvec{z}_{l,kt}}{} &= \frac{1}{{Z}_{{l,kt}}}\sum_{x_{kt}\in\tilde{\mathcal{X}}}\phi(x_{kt})\cdot\big(\Cmsga{\lvec{z}_{l,kt}}{}(x_{kt})\\
    &\qquad+\Mumsga{\lvec{z}_{l,kt}}{}(x_{kt})\cdot\Mumsga{\lvec{z}_{l,kt}}{}\herm(x_{kt})\big) - \Mumsgb{\lvec{z}_{l,kt}}{}\Mumsgb{\lvec{z}_{l,kt}}{}\herm,
    \end{split}
    \label{eq:C_mmd_Psi_z_z}
\end{align}
with $\tilde{\mathcal{X}}=\{\pilot{x}_{kt}\}$ for $t\leq T_p$, $\tilde{\mathcal{X}}=\mathcal{X}$ for $t>T_p$, and the other parameters defined in~\eqref{eq:phi_mmd_Psi_z}-\eqref{eq:Z_Psi_z}.
Hence, the updates for $t\leq T_p$ simplify to $\Mumsgb{\lvec{z}_{l,kt}}{}\!\!=\!\Mumsga{\lvec{z}_{l,kt}}{}(\pilot{x}_{kt})$ and $\Cmsgb{\lvec{z}_{l,kt}}{}\!\!=\!\Cmsga{\lvec{z}_{l,kt}}{}(\pilot{x}_{kt})$.

Furthermore, to improve numerical stability and convergence, damping is applied in the updating of all factor-to-variable messages using a parameter $\eta\in[0,1]$ as in~\cite{Karataev2024,Ngo2020}: the updated parameter is a convex combination of its value from the previous iteration and the newly computed  value at the factor node.
Finally,  messages  involving a Gaussian component are updated only if the resulting matrix parameter is positive definite; otherwise, the corresponding parameters from the previous iteration are retained.

\vspace*{-2mm}
\subsection{Final Estimation}\label{subsec:estimation}
From the final \ac{EP} messages, the approximate \ac{APP} distributions  of the user activities, channels, and data are obtained as the product of all the incoming messages at the corresponding variable nodes. The final estimates are then given by the \ac{MAP} solution of these \acp{APP},
\vspace*{-2mm}
\begin{align}
	\hat{u}_k &= \argmax_{u_k\in\{0,1\}}\;\tilde{p}_{u_k}(u_k)\cdot\prod_{l=1}^L\catmsg{\Psi_{g_{l,k}}}{u_{k}}(u_k),
	\label{eq:estimate_u}\\
        \hat{\lvec{h}}_{l,k} &= \argmax_{\lvec{h}_{l,k}\in\Cset^N}\;\tilde{p}_{h_{l,k}}(\lvec{h}_{l,k})\cdot\msgp{\Psi_{g_{l,k}}}{\lvec{h}_{l,k}}(\lvec{h}_{l,k}),
	\label{eq:estimate_h}\\
        \hat{x}_{kt} &= \argmax_{x_{kt}\in\mathcal{X}}\;\prod_{l=1}^L\catmsg{\Psi_{z_{l,kt}}}{x_{kt}}(x_{kt})\qquad{\text{for }t>T_p}.
	\label{eq:estimate_x}
\end{align}
Assuming that \ac{UE} $k$ is active, the solution to~\eqref{eq:estimate_h}  is Gaussian, yielding the following closed-form estimate:
\begin{equation}
    \hat{\lvec{h}}_{l,k} = \hat{\lmat{\Lambda}}_{\lvec{h}_{l,k}}^{-1}\hat{\lvec{\gamma}}_{\lvec{h}_{l,k}},
    \label{eq:estimate_h_explicit}
\end{equation}
with
\vspace*{-2mm}
\begin{align}
    \hat{\lvec{\gamma}}_{\lvec{h}_{l,k}} &= \tilde{\gvec{\gamma}}_{h_{l,k}}+\Gammamsg{\lvec{g}_{l,k}}{\Psi_{g_{l,k}}},
    \label{eq:posterior_h_gamma}\\
    \hat{\lmat{\Lambda}}_{\lvec{h}_{l,k}} &= \tilde{\gmat{\Lambda}}_{h_{l,k}}+\Lambdamsg{\lvec{g}_{l,k}}{\Psi_{g_{l,k}}}.
    \label{eq:posterior_h_Lambda}
\end{align}

\vspace*{-2mm}
\subsection{Computational Complexity and Scalability}\label{subsec:complexity}
In the following, we characterize the  computational complexity order of the  proposed algorithms neglecting  addition and subtraction costs.
At the \acp{AP}, the dominant computational burden for both algorithms arises from  the evaluation of the Gaussian likelihood in~\eqref{eq:theta_mmd_Psi_z_z}, which requires computing the inverse and determinant of  $N-$dimensional covariance matrices.
This computation is performed $KT_d$ times for every constellation symbol $x_{kt}\in\mathcal{X}$ and exhibits a complexity of $\mathcal{O}(N^3)$.
Consequently, the per-iteration computational complexity is $\mathcal{O}(M N^3 K T_d)$  at every \ac{AP}.
From a network scalability perspective, the key observation is that the computational load at the \acp{AP}  scales linearly with the number of \acp{UE} $K$.
This dependency can be strongly mitigated by pruning \acp{UE} with negligible large-scale fading coefficients, i.e., \acp{UE} that are unlikely to contribute meaningfully to the received signal.
Such a modification reduces the effective number of \acp{UE} processed at each \ac{AP} without altering the core algorithmic structure.

The complexity at the \ac{CPU} for both algorithms is mainly determined by the combination of symbol beliefs in~\eqref{eq:m_x_Psi_z} which has a complexity order of $\mathcal{O}(LMKT_d)$.
Hence, the \ac{CPU} complexity scales linearly with both the number of \acp{UE} and \acp{AP}, which may limit scalability in very large networks.
Nevertheless, the processing of different symbols at the \ac{CPU} is fully decoupled, enabling straightforward parallel and/or distributed implementations that can effectively remove this scalability bottleneck.

\vspace*{-1mm}
\section{Numerical Results}\label{sec:sims}
\vspace*{-1mm}
\begin{figure*}[t]
    \centering
    \subfloat[\acs{DER}.]{\includegraphics[width=\plotwidth\textwidth]{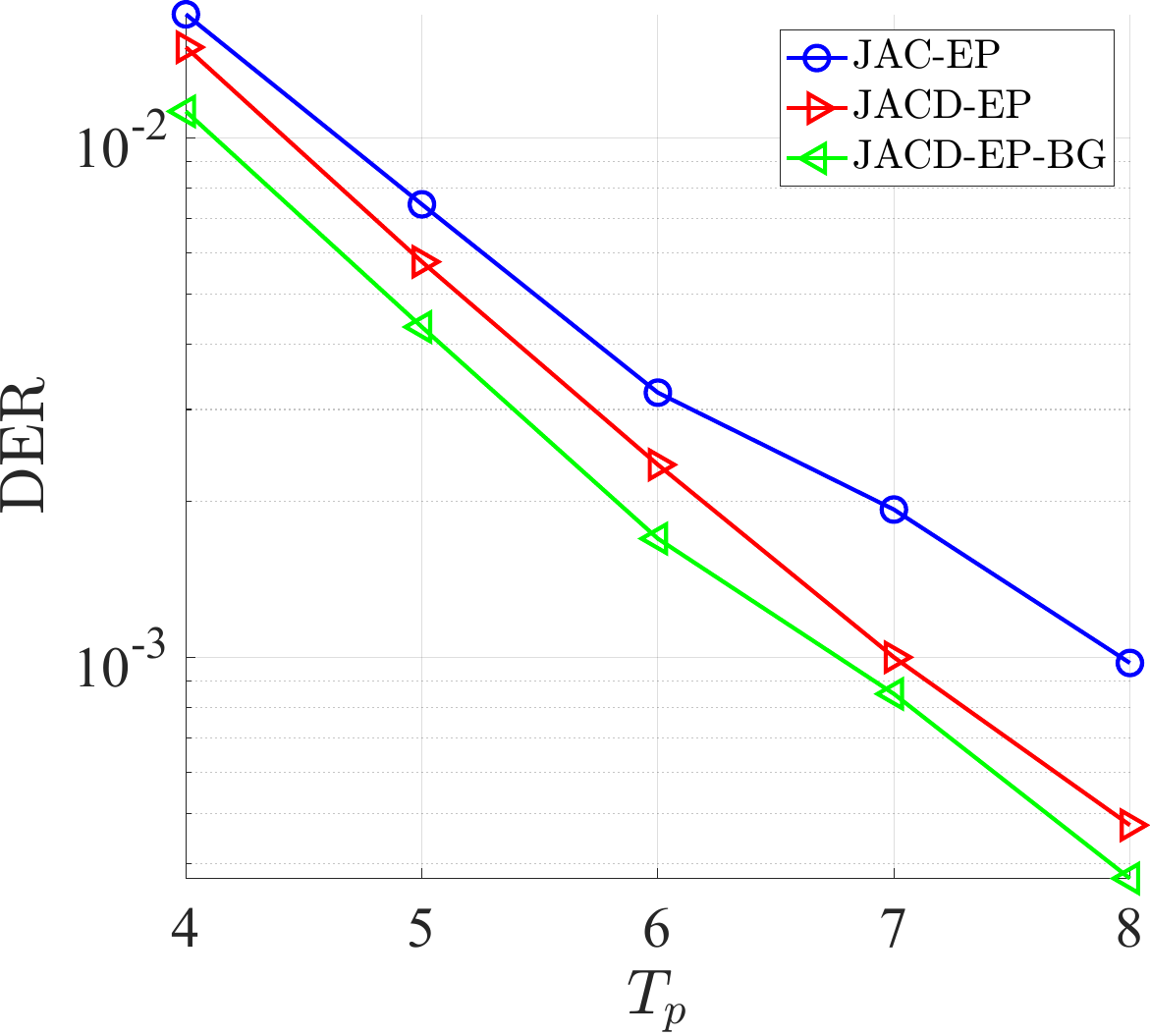}
        \label{fig:DER_Tp}}\hfill
    \subfloat[\acs{NMSE}.]{\includegraphics[width=\plotwidth\textwidth]{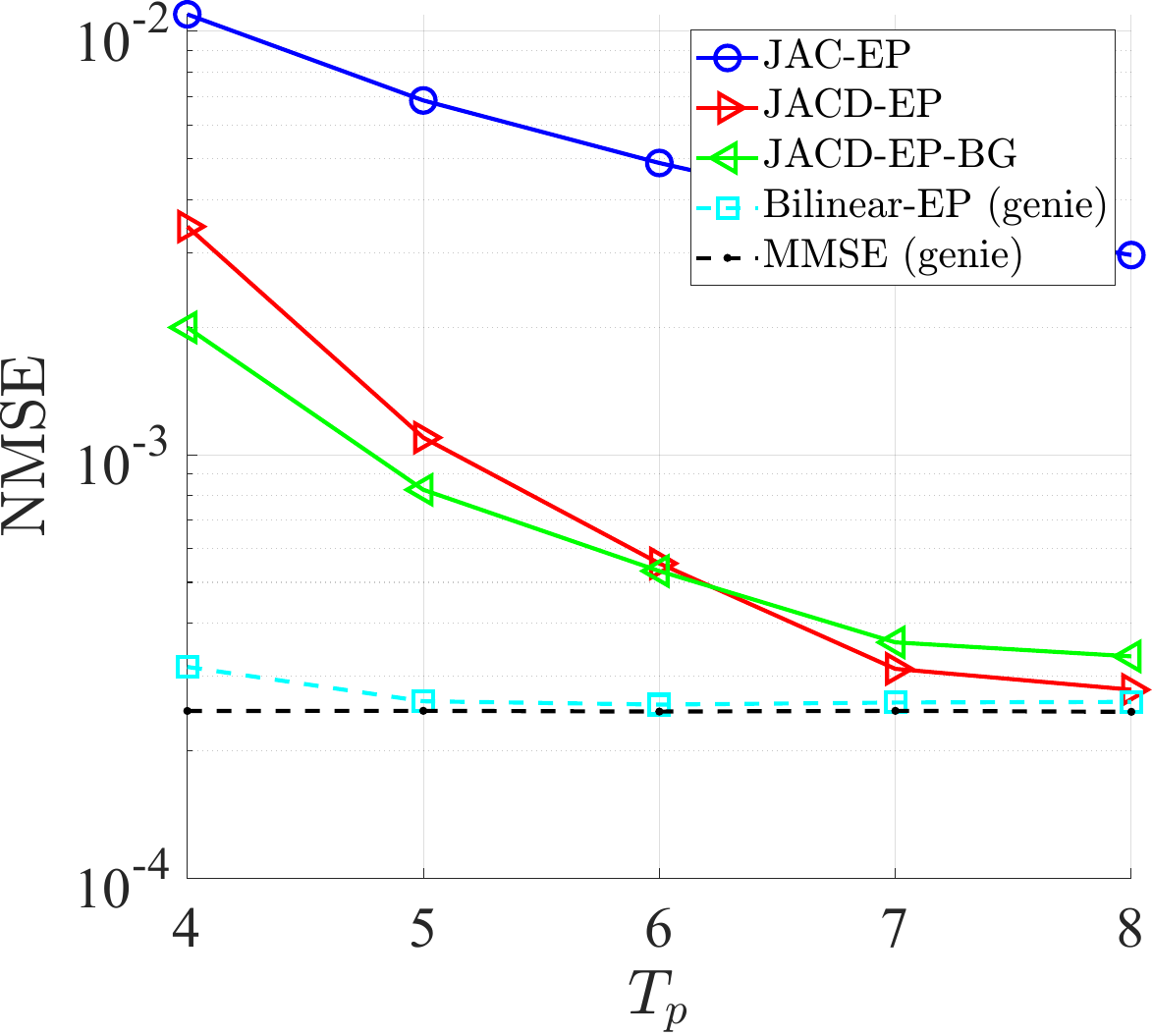}
        \label{fig:NMSE_Tp}}\hfill
    \subfloat[\acs{SER}.]{\includegraphics[width=\plotwidth\textwidth]{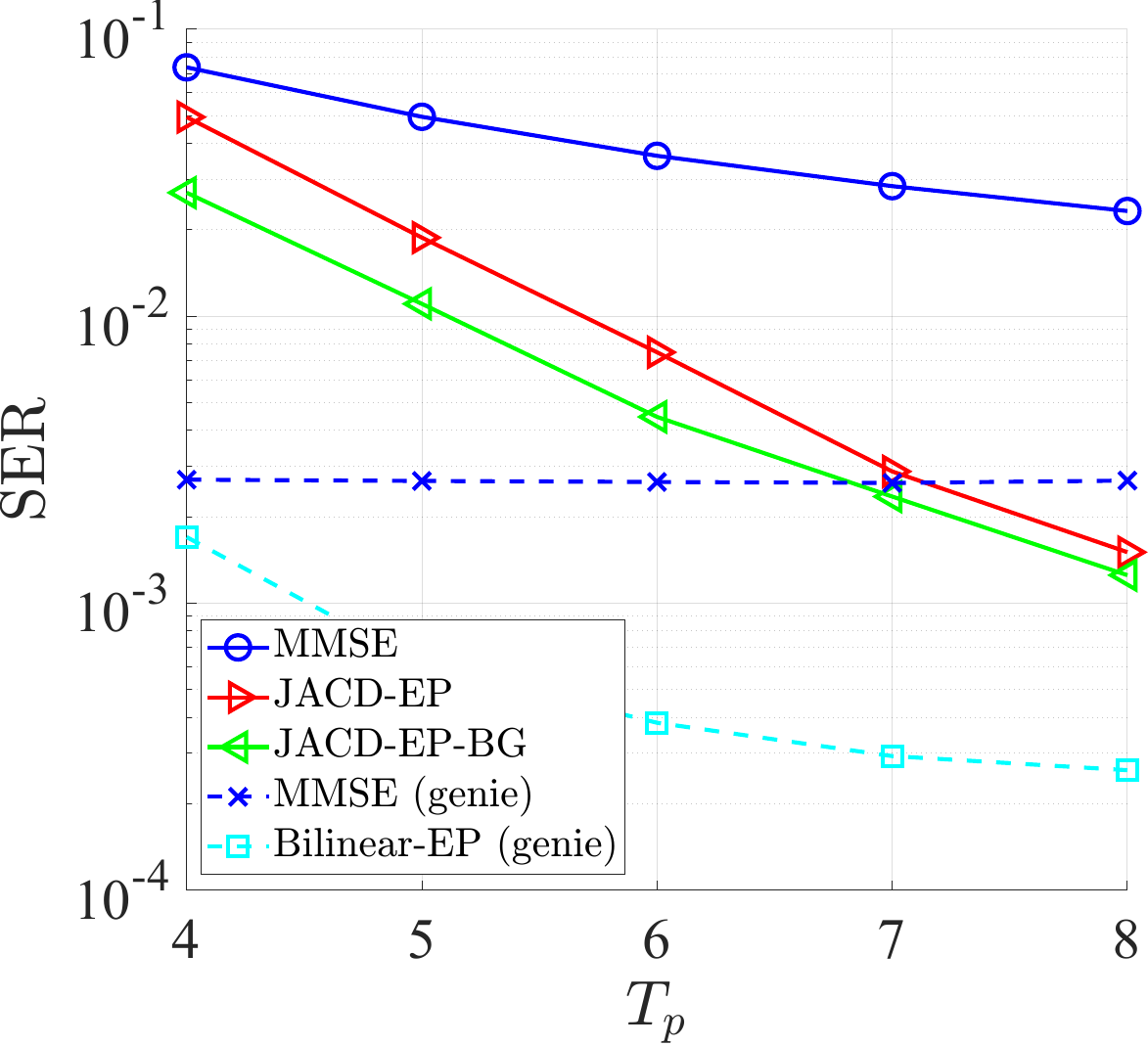}
        \label{fig:SER_Tp}}
    \caption{Performance metrics versus pilot sequence length for $L=25$, $N=1$, $K=40$, $\lambda=0.3$, and $T=60$.}
    \label{fig:results_Tp}
    \vspace*{-3mm}
\end{figure*}
In this section, we evaluate the performance of the proposed \ac{EP}-based \ac{JACD} algorithms via extensive Monte Carlo simulations.
We consider a square network area of $500\times500\,$m$^2$ comprising $L=25$ single-antenna \acp{AP}, i.e., $N=1$, placed on a uniform grid at the positions $\{(50+i\!\cdot\!100,50+j\!\cdot\!100)\,\text{m}\,|\,i,j\in\{0,1,2,3,4\}\}$ and a height of 10 m.
The receiver noise power at each \ac{AP} is set to $\sigma_n^2=-96\,$dBm.
A total of $K=40$ \acp{UE} are uniformly and independently placed at random ground locations, each with an activity probability $\lambda=0.3$.
An active \ac{UE} transmits $T_p$ pilot symbols and $T_d=T-T_p$ 4-\ac{QAM} data symbols with constant transmit power $\sigma_x^2=16\,$dBm.
The $T_p$ columns of the pilot matrix $\pilot{\lmat{X}}$ are chosen from the $K\times K$ \ac{DFT} matrix such that the maximum inner product between any two different pilot sequences is minimized, thereby minimizing the mutual coherence~\cite{Rusu2018}.
Due to the limited pilot sequence length and number of \acp{UE} in our simulations, the optimal pilot sequences can be found by a full search.
For larger systems, optimization-based pilot designs such as those proposed in~\cite{Rusu2018,Iimori2021} would be required.
The large-scale fading coefficients follow the 3GPP urban microcell model that incorporates correlated shadow fading~\cite[Sec.~2.5.2]{Demir2021}.
Furthermore, we assume a channel coherence time of $T=60$ which corresponds to $16\,$ms using a subcarrier spacing of $3.75\,$kHz as defined in narrowband-\ac{IoT} over 5G~\cite{Chen2017}.
This choice aligns with the channel coherence time of practical channels, especially low mobility use cases, and, at the same time, limits the number of jointly processed symbols, hence, reducing overall computational complexity and latency.
In practice, the number of jointly processed symbols can be chosen arbitrarily according to the complexity and delay requirements, provided that all processed symbols remain within the same channel coherence interval.
Note that in the simulation we implement sequentially an algorithm that should be implemented in parallel in a distributed setting.
Furthermore, the computations within the \ac{CPU} and each \ac{AP} can be parallelized as well.
Hence, the algorithm is still scalable even for larger systems which are, however, not considered here for limiting the simulation complexity.

Both the \ac{JACD-EP} and \ac{JACD-EP-BG} algorithm use the \ac{JAC-EP} algorithm from~\cite{Forsch2024} for pilot-based initialization.
For benchmarking, we consider the centralized linear \ac{MMSE} \ac{MIMO} detector~\cite{Bjoernson2020} using (1) imperfect activity and channel information obtained from the \ac{JAC-EP} algorithm and (2) genie-aided perfect activity and channel knowledge.
Additionally, we consider the genie-aided \ac{MMSE} channel estimator with perfectly known data symbols and the bilinear-\ac{EP} algorithm~\cite{Forsch2025} with perfect \ac{UE} activity knowledge and pilot-based \ac{MMSE} channel estimation initialization as performance upper bounds.
All the iterative \ac{EP}-based algorithms employ a damping parameter of $\eta=0.5$ and perform 20 iterations.

\begin{figure*}[t]
    \centering
    \subfloat[\acs{DER}.]{\includegraphics[width=\plotwidth\textwidth]{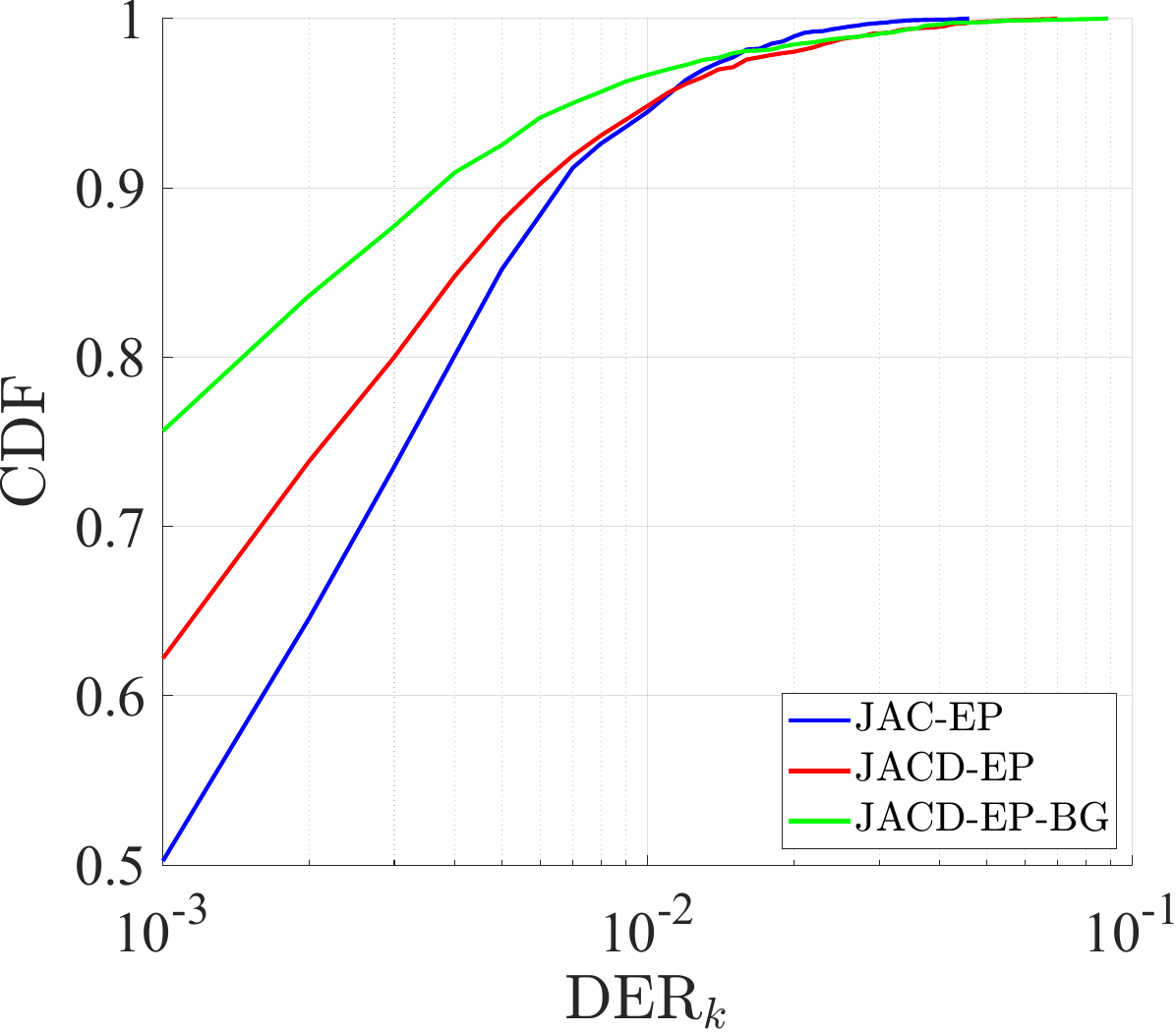}
        \label{fig:DER_CDF}}\hfill
    \subfloat[\acs{NMSE}.]{\includegraphics[width=\plotwidth\textwidth]{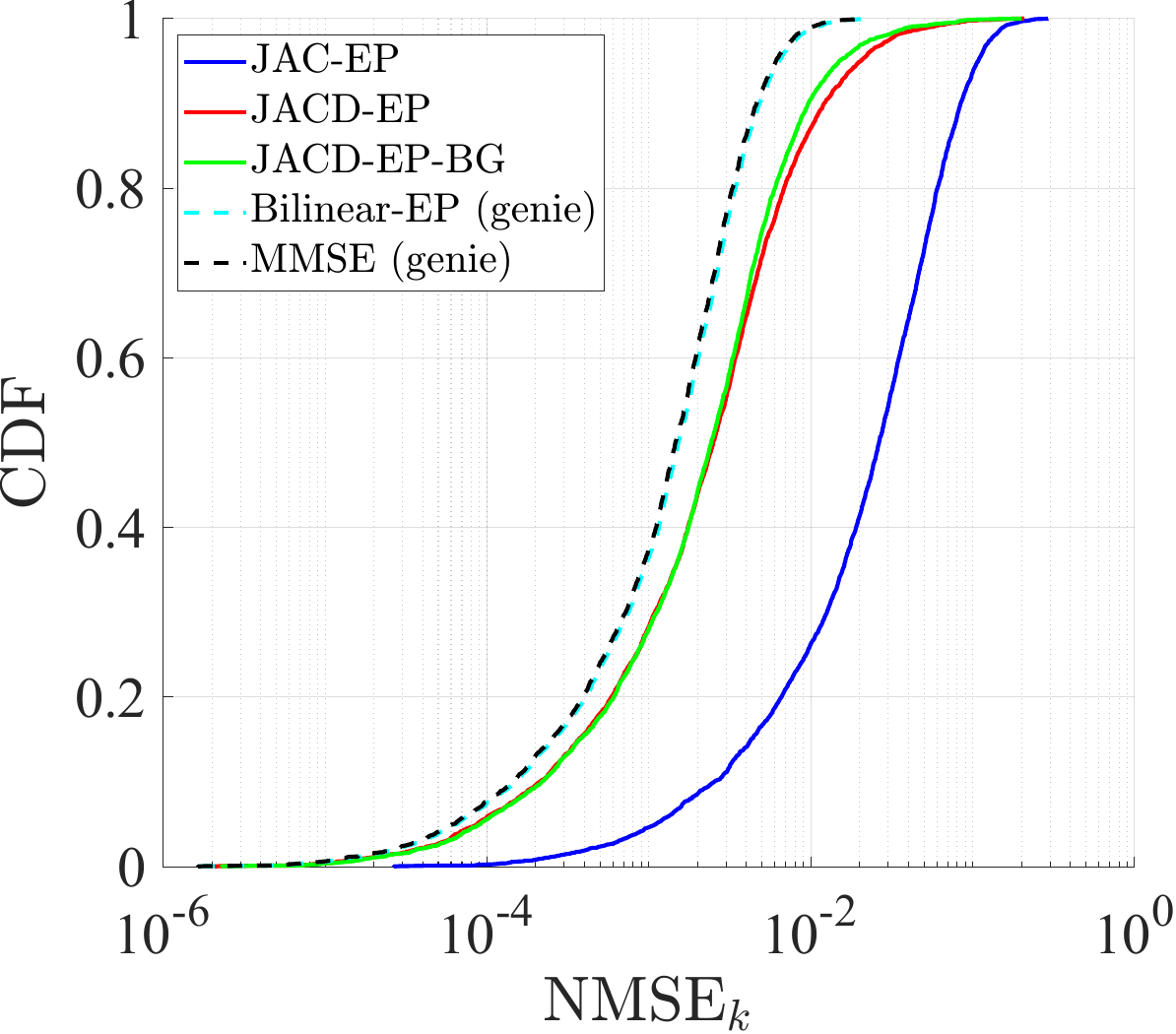}
        \label{fig:NMSE_CDF}}\hfill
    \subfloat[\acs{SER}.]{\includegraphics[width=\plotwidth\textwidth]{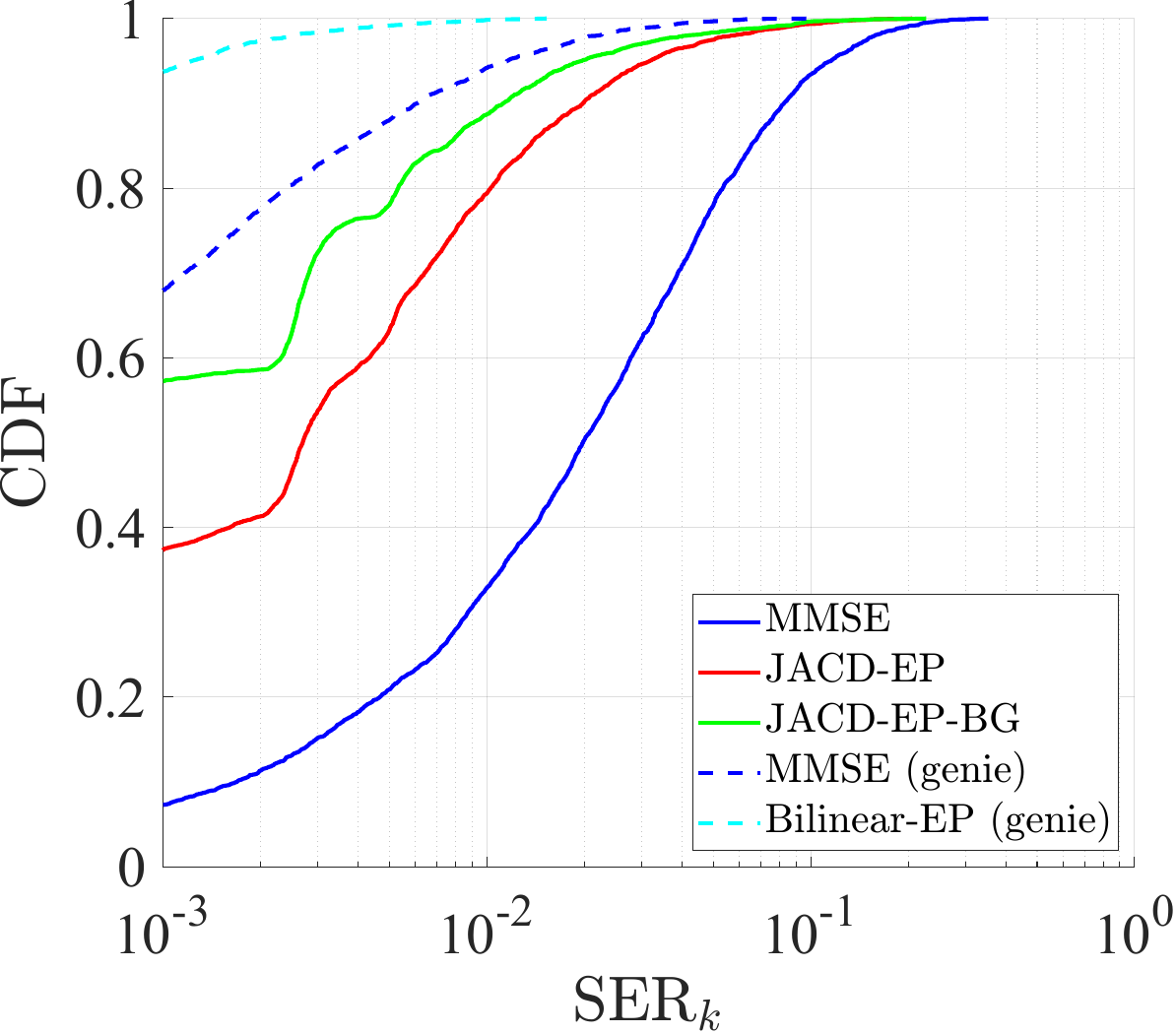}
        \label{fig:SER_CDF}}
    \caption{\acsp{CDF} of performance metrics for $L=25$, $N=1$, $K=40$, $\lambda=0.3$, $T=60$, and $T_p=6$.}
    \label{fig:results_CDF}
    \vspace*{-3mm}
\end{figure*}

The performance is assessed in terms of the \ac{DER} for \ac{UE} activity detection, the \ac{NMSE} for channel estimation, and the \ac{SER} for data detection.
The performance metrics are averaged across all \acp{UE}.
The \ac{DER}, which incorporates misdetections and false alarms, is computed as $\text{DER}\coloneq\est{}{\sum_k\ind{u_k\neq\hat{u}_k} / K}$, the \ac{NMSE} of the channel estimate $\hat{\lmat{G}}=\hat{\lmat{H}}\hat{\lmat{U}}$ is defined as $\text{NMSE}\coloneq\frac{\est{}{||\lmat{G}-\hat{\lmat{G}}||_F^2}}{\est{}{||\lmat{G}||_F^2}}$ with $\lmat{G}=\lmat{H}\lmat{U}$, and the \ac{SER} is given by $\text{SER}\coloneq\est{}{\sum_k\sum_t\ind{\data{x}_{kt}\neq\data{\hat{x}}_{kt}} / (KT_d)\,|\,u_k=1}$.
Note that as active \acp{UE} are erroneously classified as inactive, the corresponding data symbol estimates are randomly chosen from the symbol constellation $\mathcal{X}$, leading to a higher \ac{SER} in case of misdetections.
The expectation operator in the definitions above is computed with respect to the \ac{UE} activity, channel, and noise realizations.
The results in Fig.~\ref{fig:results_Tp} are obtained by averaging over $10^3$ block transmissions, each with independent \ac{UE} positions and activities.
Fig.~\ref{fig:results_Tp} illustrates the three performance metrics as a function of the pilot sequence length $T_p$, which reflects the strength of \ac{PC}.
The proposed algorithms outperform the \ac{JAC-EP} initialization algorithm in terms of \ac{DER} and \ac{NMSE}.
The \ac{DER} performance enhances steadily with increasing $T_p$, whereas the \ac{NMSE} performance saturates and approaches the \ac{MMSE} lower bound for $T_p>6$ due to the fixed coherence block length $T=60$.
Furthermore, the proposed \ac{EP}-based algorithms significantly outperform the linear \ac{MMSE} MIMO detector in terms of the \ac{SER}.
The proposed algorithms reach the same performance and even outperform the genie-aided linear \ac{MMSE} MIMO detector with perfect activity and channel knowledge for $T_p>6$.
A noticeable performance gap remains with respect to the genie-aided bilinear-\ac{EP} algorithm due to the challenging \ac{UE} activity detection under \ac{PC}.
Hence, we can conclude that the overall system performance is limited by the activity detection capability.
Furthermore, the \ac{JACD-EP-BG} algorithm outperforms the \ac{JACD-EP} algorithm for small $T_p$ which corresponds to strong \ac{PC}.

Next, we investigate the \ac{QoS} distribution within the network by evaluating the empirical \acp{CDF} of the \ac{DER}, \ac{NMSE}, and \ac{SER} in Fig.~\ref{fig:results_CDF}.
The three metrics are computed per \ac{UE}, i.e., $\text{DER}_k\coloneq\est{}{\ind{u_k\neq\hat{u}_k}}$, $\text{NMSE}_k\coloneq\frac{\est{}{||\lvec{g}_k-\hat{\lvec{g}}_k||^2}}{\est{}{||\lvec{g}||^2}}$ with $\hat{\lvec{g}}_k$ and $\lvec{g}_k$ being the $k^\text{th}$ column of $\hat{\lmat{G}}$ and $\lmat{G}$, respectively, and $\text{SER}_k\coloneq\est{}{\sum_t\ind{\data{x}_{kt}\neq\data{\hat{x}}_{kt}} / T_d\,|\,u_k=1}$.
The \acp{CDF} are obtained by considering 100 different realizations of the \ac{UE} positions, resulting in $100\cdot K=4000$ data points.
For each realization, the metrics are averaged over $10^3$ independent block transmissions, accounting for \ac{UE} activity, small-scale fading, and noise realizations.
Pilot sequences of length $T_p=6$ are used.
Fig.~\ref{fig:results_CDF} shows again the superior performance of the proposed algorithms.
Furthermore, it can be observed that the \ac{JACD-EP-BG} algorithm outperforms the \ac{JACD-EP} algorithm, especially for \acp{UE} experiencing good propagation conditions.

\vspace*{-2mm}
\section{Conclusion}\label{sec:concl}
\vspace*{-1mm}
In this paper, we considered the uplink of a \ac{GF-CF-MaMIMO} system and tackled the \ac{JACD} problem for massive \ac{MTC} under \ac{PC}.
We developed an \ac{EP}-based framework for distributed and scalable \ac{JACD} by appropriately factorizing the \ac{APP} distribution and applying \ac{EP} message passing on the associated factor graph.
Within this framework, we employed accurate categorical probability distributions for user activities and data, and Gaussian or \ac{BG} probability distributions for the channels, resulting in the \ac{JACD-EP} and the \ac{JACD-EP-BG} algorithms, respectively.
To support the latter, we developed an exponential family representation of the \ac{BG} distribution. The proposed framework is inherently scalable with respect to both the number of \acp{AP} and \acp{UE}. The message-passing structure enables straightforward extensions toward fully distributed implementations, while the computational complexity at each \ac{AP} scales linearly with the number of \acp{UE} in its coverage area.
The numerical results showed that the proposed algorithms exhibit strong robustness against \ac{PC} and outperform state-of-the-art algorithms in terms of \ac{DER}, \ac{NMSE}, and \ac{SER}.

\appendices
\vspace*{-2mm}
\section{Proof of the Bernoulli-Gaussian Product Lemma}\label{app:BG_product}
Using the exponential family representation of \ac{BG} distributions, $\BG{\lvec{x}|\lambda_i,\gvec{\mu}_i,\lmat{C}_i} = \er^{\gvec{\eta}_{\mathrm{BG},i}\herm\lvec{u}_\mathrm{BG}(\lvec{x}) - A_\mathrm{BG}(\gvec{\eta}_{\mathrm{BG},i})}$, $i=\{1,2\}$, the product of two \ac{BG} distributions is given by
\begin{align}
	&\BG{\lvec{x}|\lambda_1,\gvec{\mu}_1,\lmat{C}_1} \cdot \BG{\lvec{x}|\lambda_2,\gvec{\mu}_2,\lmat{C}_2}\nonumber\\
	&\quad= \er^{\left(\gvec{\eta}_{\text{BG},1}+\gvec{\eta}_{\text{BG},2}\right)\herm\lvec{u}_\text{BG}(\lvec{x}) - A_\text{BG}(\gvec{\eta}_{\text{BG},1}+\gvec{\eta}_{\text{BG},2})}\nonumber\\
    &\qquad\cdot \er^{A_\text{BG}(\gvec{\eta}_{\text{BG},1}+\gvec{\eta}_{\text{BG},2})-A_\text{BG}(\gvec{\eta}_{\text{BG},1})-A_\text{BG}(\gvec{\eta}_{\text{BG},2})}\nonumber\\
	&\quad= \BG{\lvec{x}|\lambda,\gvec{\mu},\lmat{C}} \cdot \er^{A_\text{BG}(\gvec{\eta}_{\text{BG},1}+\gvec{\eta}_{\text{BG},2})-A_\text{BG}(\gvec{\eta}_{\text{BG},1})-A_\text{BG}(\gvec{\eta}_{\text{BG},2})}.
	\label{eq:app_BG_product}
\end{align}
The first factor in~\eqref{eq:app_BG_product} is a normalized \ac{BG} distribution with vector of natural parameters $\gvec{\eta}_\text{BG}=\gvec{\eta}_{\text{BG},1}+\gvec{\eta}_{\text{BG},2}$, or equivalently, $\kappa=\kappa_1+\kappa_2$, $\gvec{\gamma}=\gvec{\gamma}_1+\gvec{\gamma}_2$, and $\gmat{\Lambda}=\gmat{\Lambda}_1+\gmat{\Lambda}_2$.
From $\gvec{\gamma}$ and $\gmat{\Lambda}$, the covariance matrix and mean of the Gaussian component in~\eqref{eq:BG_product_C} and~\eqref{eq:BG_product_mu}, respectively, immediately follows.
Recall that $A_\text{G}(\gvec{\eta}_\text{G})=\log\frac{1}{\er^{-A_\text{G}(\gvec{\eta}_\text{G})}}=\log\frac{1}{\CN{\lvec{0}|\gvec{\mu},\lmat{C}}}$.
Now, using the Gaussian product lemma and denoting $\gvec{\eta}_\text{G}=\gvec{\eta}_{\text{G},1}+\gvec{\eta}_{\text{G},2}$, the first natural parameter can be rewritten as
\begin{align}
	\kappa &= \kappa_1 + \kappa_2\nonumber\\
	&= \log\Big(\frac{(1-\lambda_1)(1-\lambda_2)}{\lambda_1\lambda_2}\!\cdot\!\frac{1}{\CN{\lvec{0}|\gvec{\mu}_1,\lmat{C}_1}\CN{\lvec{0}|\gvec{\mu}_2,\lmat{C}_2}}\Big)\nonumber\\
	&= \log\Big(\frac{(1-\lambda_1)(1-\lambda_2)}{\lambda_1\lambda_2}\nonumber\\
    &\qquad\qquad\cdot\frac{1}{\CN{\lvec{0}|\gvec{\mu},\lmat{C}}\CN{\lvec{0}|\gvec{\mu}_1-\gvec{\mu}_2,\lmat{C}_1+\lmat{C}_2}}\Big)\nonumber\\
    \begin{split}
	&= \log\Big(\frac{(1\!-\!\lambda_1)(1\!-\!\lambda_2)}{\lambda_1\lambda_2}\!\cdot\!\frac{1}{\CN{\lvec{0}|\gvec{\mu}_1\!-\!\gvec{\mu}_2,\lmat{C}_1\!+\!\lmat{C}_2}}\Big)\\
    &\quad+ A_\text{G}(\gvec{\eta}_\text{G}),
    \end{split}
    \label{eq:app_BG_product_kappa}
\end{align}
where $\lmat{C}$ and $\gvec{\mu}$ are given by~\eqref{eq:BG_product_C} and~\eqref{eq:BG_product_mu}, respectively.
Using the relation $\lambda=\frac{1}{1+\er^{\kappa-A_\text{G}(\gvec{\eta}_\text{G})}}$ together with~\eqref{eq:app_BG_product_kappa} yields~\eqref{eq:BG_product_lambda}.
Noting that $\er^{A_\text{BG}(\gvec{\eta}_\text{BG})}=\er^{A_\text{G}(\gvec{\eta}_\text{G})-\log\lambda}={\er^{A_\text{G}(\gvec{\eta}_\text{G})}\cdot\frac{1}{\lambda}}=\frac{1}{\lambda\cdot\CN{\lvec{0}|\gvec{\mu},\lmat{C}}}$ and applying the Gaussian product lemma and~\eqref{eq:BG_product_lambda}, it can be shown that the normalization constant for the \ac{BG} product given in~\eqref{eq:app_BG_product} is equal to the normalization constant stated in~\eqref{eq:BG_product}.

\vspace*{-2mm}
\section{Derivation of Message-Passing Update Rules}\label{app:MP_updates}
\vspace*{-1mm}
In the following, we briefly present the general message-passing update rules for \ac{EP} on graphs.
Next, we derive the message-passing rules for the \ac{JACD-EP-BG} algorithm presented in Section~\ref{subsec:MP_updates}.
For brevity, we focus on factor-to-variable messages.
Variable-to-factor messages can be easily derived by using the corresponding general \ac{EP} message-passing update rule in~\eqref{eq:EP_var_to_fac_message} and the fact that~\eqref{eq:EP_var_to_fac_message} is applied to exponential family distributions, which makes the computation of products straightforward.
Detailed derivations of the message-passing rules utilized in the \ac{JACD-EP} algorithm can be found in the extended version of~\cite{Forsch2024}.

\vspace*{-2mm}
\subsection{Expectation Propagation on Graphs}\label{subsec:EP_graph}
Consider a factor graph with factor nodes $\Psi_\alpha$ and variable nodes $\lvec{x}_\beta$.
Let $\lvec{x}_\alpha$ be the vector containing all variables connected to $\Psi_\alpha$.
Let $N_\beta$ denote the set of neighboring factor node indices of $\lvec{x}_\beta$, i.e., $N_\beta=\{\alpha\,|\,\lvec{x}_\beta\subseteq\lvec{x}_\alpha\}$.
The variable-to-factor message $\msg{\lvec{x}_\beta}{\Psi_\alpha}$ is obtained by computing the parameters of the following distribution,
\begin{equation}
    \msgp{\lvec{x}_\beta}{\Psi_\alpha}(\lvec{x}_\beta) \propto \prod_{\alpha'\in N_\beta\setminus\alpha}\msgp{\Psi_{\alpha'}}{\lvec{x}_\beta}(\lvec{x}_\beta).
	\label{eq:EP_var_to_fac_message}
\end{equation}
The factor-to-variable message $\msg{\Psi_\alpha}{\lvec{x}_\beta}$ is computed by determining the parameters of
\vspace*{-2mm}
\begin{equation}
	\msgp{\Psi_\alpha}{\lvec{x}_\beta}(\lvec{x}_\beta) \propto \frac{\text{proj}\left\{\mmd{\Psi_{\alpha}}{\lvec{x}_{\beta}}(\lvec{x}_\beta)\right\}}{\msgp{\lvec{x}_\beta}{\Psi_\alpha}(\lvec{x}_\beta)},
	\label{eq:EP_fac_to_var_message}
\end{equation}
where the distribution $\mmd{\Psi_{\alpha}}{\lvec{x}_{\beta}}(\lvec{x}_\beta)$ is given by
\begin{equation}
	\mmd{\Psi_{\alpha}}{\lvec{x}_{\beta}}(\lvec{x}_\beta) = \frac{1}{{Z}_{\Psi_\alpha}}\int\Psi_\alpha(\lvec{x}_\alpha)\prod_{\beta'\in N_\alpha}\msgp{\lvec{x}_{\beta'}}{\Psi_\alpha}(\lvec{x}_{\beta'})\,\mathrm{d}\lvec{x}_\alpha\!\setminus\!\lvec{x}_\beta.
	\label{eq:message_projection}
\end{equation}
The distribution $\mmd{\Psi_{\alpha}}{\lvec{x}_{\beta}}(\lvec{x}_\beta)$ is also referred to as the \emph{local belief of $\lvec{x}_{\beta}$ at the factor node $\Psi_{\alpha}$}.
Here, ${Z}_{\Psi_\alpha}$ is a normalization constant and $N_\alpha$ denotes the set of neighboring variable node indices of $\Psi_\alpha$, i.e, $N_\alpha=\{\beta\,|\,\lvec{x}_\beta\subseteq\lvec{x}_\alpha\}$.
Furthermore, $\text{proj}\{\cdot\}$ denotes the projection onto the chosen exponential family defined as
\begin{equation}
	\text{proj}\{f(\lvec{x})\} = \argmin_{g(\lvec{x})\in\mathcal{F}}D_{KL}(f(\lvec{x})||g(\lvec{x})),
	\label{eq:projection}
\end{equation}
where $D_{KL}(f(\lvec{x})||g(\lvec{x}))$ is the \ac{KL} divergence between $f$ and $g$ and $\mathcal{F}$ is the exponential family with sufficient statistics $\lvec{u}(\lvec{x})$.
This minimization is performed by \emph{moment matching},
\vspace*{-1mm}
\begin{align}
    \est{g(\lvec{x})}{\lvec{u}(\lvec{x})} = \est{f(\lvec{x})}{\lvec{u}(\lvec{x})}.
    \label{eq:moment_matching}
\end{align}
Once message passing converges, the approximate posterior distribution $\hat{p}_{\lvec{x}_\beta}(\lvec{x}_\beta)$ of the variable $\lvec{x}_\beta$ can be computed via
\vspace*{-1mm}
\begin{align}
	\hat{p}_{\lvec{x}_\beta}(\lvec{x}_\beta) &\propto \prod_{\alpha\in N_\beta} \msgp{\Psi_\alpha}{\lvec{x}_\beta}(\lvec{x}_\beta).
	\label{eq:approx_factor_beta_factorization}
\end{align}

\vspace*{-6mm}
\subsection{Message Update for $\msg{\Psi_{y_{l,t}}}{\lvec{z}_{l,kt}}$}\label{subsec:m_Psi_y_z}
This message update is identical for the \ac{JACD-EP-BG} and the \ac{JACD-EP} algorithm, and, hence, its derivation can be found in the extended version of~\cite{Forsch2024}.

\vspace*{-2mm}
\subsection{Message Update for $\msg{\Psi_{z_{l,kt}}}{x_{kt}}$}\label{subsec:m_Psi_z_x}
The local belief of $x_{kt}$ at the factor node $\Psi_{z_{l,kt}}$ is given by
\begin{align}
	&\mmd{\Psi_{z_{l,kt}}}{x_{kt}}(x_{kt})\nonumber\\
    &\quad\propto \int\!\!\int\delta(\lvec{z}_{l,kt}-\lvec{g}_{l,k}x_{kt})\cdot\msgp{x_{kt}}{\Psi_{z_{l,kt}}}(x_{kt})\nonumber\\
	&\qquad\quad\cdot\msgp{\lvec{z}_{l,kt}}{\Psi_{z_{l,kt}}}(\lvec{z}_{l,kt})\cdot\msgp{\lvec{g}_{l,k}}{\Psi_{z_{l,kt}}}(\lvec{g}_{l,k})\,d\lvec{z}_{l,kt}\,d\lvec{g}_{l,k}\nonumber\\
	&\quad\overset{(a)}{=} \msgp{x_{kt}}{\Psi_{z_{l,kt}}}(x_{kt})\cdot\!\!\int\msgp{\lvec{z}_{l,kt}}{\Psi_{z_{l,kt}}}(\lvec{g}_{l,k}x_{kt})\nonumber\\
    &\qquad\cdot\msgp{\lvec{g}_{l,k}}{\Psi_{z_{l,kt}}}(\lvec{g}_{l,k})\,d\lvec{g}_{l,k}\nonumber\\
	&\quad\overset{(b)}{\approx} \msgp{x_{kt}}{\Psi_{z_{l,kt}}}(x_{kt})\cdot\!\!\int\mathcal{CN}\big(\lvec{g}_{l,k}x_{kt}|\Mumsg{\Psi_{y_{l,t}}}{\lvec{z}_{l,kt}},\Cmsg{\Psi_{y_{l,t}}}{\lvec{z}_{l,kt}}\big)\nonumber\\
	&\qquad\cdot\Actmsg{\lvec{g}_{l,k}}{\Psi_{z_{l,kt}}}\cdot\mathcal{CN}\big(\lvec{g}_{l,k}|\Mumsg{\lvec{g}_{l,k}}{\Psi_{z_{l,kt}}},\Cmsg{\lvec{g}_{l,k}}{\Psi_{z_{l,kt}}}\big)\,d\lvec{g}_{l,k}\nonumber\\
	&\quad\overset{(c)}{=} \msgp{x_{kt}}{\Psi_{z_{l,kt}}}(x_{kt})\cdot\!\!\int|x_{kt}|^{-2N}\cdot\Actmsg{\lvec{g}_{l,k}}{\Psi_{z_{l,kt}}}\nonumber\\
	&\qquad\cdot\CN{\lvec{g}_{l,k}|\gvec{\mu}_\text{tmp},\lmat{C}_\text{tmp}}\cdot\mathcal{CN}\big(\lvec{0}|\Mumsg{\Psi_{y_{l,t}}}{\lvec{z}_{l,kt}}x_{kt}^{-1}\nonumber\\
    &\qquad-\Mumsg{\lvec{g}_{l,k}}{\Psi_{z_{l,kt}}},    \Cmsg{\Psi_{y_{l,t}}}{\lvec{z}_{l,kt}}|x_{kt}|^{-2}+\Cmsg{\lvec{g}_{l,k}}{\Psi_{z_{l,kt}}}\big)\,d\lvec{g}_{l,k}\nonumber\\
	&\quad= \msgp{x_{kt}}{\Psi_{z_{l,kt}}}(x_{kt})\cdot|x_{kt}|^{-2N}\cdot\Actmsg{\lvec{g}_{l,k}}{\Psi_{z_{l,kt}}}\nonumber\\
	&\qquad\cdot\mathcal{CN}\big(\lvec{0}|\Mumsg{\Psi_{y_{l,t}}}{\lvec{z}_{l,kt}}x_{kt}^{-1}-\Mumsg{\lvec{g}_{l,k}}{\Psi_{z_{l,kt}}},\nonumber\\
    &\qquad\qquad\quad\;\;\Cmsg{\Psi_{y_{l,t}}}{\lvec{z}_{l,kt}}|x_{kt}|^{-2}+\Cmsg{\lvec{g}_{l,k}}{\Psi_{z_{l,kt}}}\big)\nonumber\\
	&\quad= \msgp{x_{kt}}{\Psi_{z_{l,kt}}}(x_{kt})\cdot\Actmsg{\lvec{g}_{l,k}}{\Psi_{z_{l,kt}}}\cdot\theta(x_{kt}),
	\label{eq:app_mmd_Psi_z_x}
\end{align}
where $(a)$ is obtained by the sifting property of the Dirac delta function~\cite{Candan2021}, $(b)$ is obtained by using $\msgp{\lvec{z}_{l,kt}}{\Psi_{z_{l,kt}}}(\lvec{z}_{l,kt}) = \msgp{\Psi_{y_{l,t}}}{\lvec{z}_{l,kt}}(\lvec{z}_{l,kt})$ and considering only the Gaussian part of the \ac{BG} distribution $\msgp{\lvec{g}_{l,k}}{\Psi_{z_{l,kt}}}$ which models the active component that is necessary to detect the transmitted symbol $x_{kt}$, and $(c)$ is obtained by utilizing the Gaussian scaling\footnote{Gaussian scaling lemma~\cite{Papoulis2002}: $\CN{\lvec{y}\,|\,c\,\gvec{\mu},|c|^2\,\lmat{C}} = |c|^{-2N}\cdot\CN{c^{-1}\,\lvec{y}|\gvec{\mu},\lmat{C}}$} and product lemma which yields $\gvec{\mu}_\text{tmp}$ and $\lmat{C}_\text{tmp}$ accordingly.
The final result in~\eqref{eq:app_mmd_Psi_z_x} is obtained by applying the Gaussian scaling lemma once again with $\theta(x_{kt})$ given in~\eqref{eq:theta_mmd_Psi_z_z}.
For $t>T_p$ and $x_{kt}\in\mathcal{X}$,~\eqref{eq:app_mmd_Psi_z_x} is a categorical distribution.
Hence, the projection operation in~\eqref{eq:EP_fac_to_var_message} is superfluous since it projects $\mmd{\Psi_{z_{l,kt}}}{x_{kt}}(x_{kt})$ onto a categorical distribution, and the denominator of~\eqref{eq:EP_fac_to_var_message} cancels with the first term in~\eqref{eq:app_mmd_Psi_z_x}.
Thus, the final message update rule reduces to~\eqref{eq:m_Psi_z_x}.

\vspace*{-2mm}
\subsection{Message Update for $\msg{\Psi_{z_{l,kt}}}{\lvec{g}_{l,k}}$}\label{subsec:m_Psi_z_g}
The local belief of $\lvec{g}_{l,k}$ at the factor node $\Psi_{z_{l,kt}}$ is given by
\vspace*{-5mm}
\begin{align}
    &\mmd{\Psi_{z_{l,kt}}}{\lvec{g}_{l,k}}(\lvec{g}_{l,k})\nonumber\\
    &\quad\propto \sum_{x_{kt}}\int\delta(\lvec{z}_{l,kt}-\lvec{g}_{l,k}x_{kt})\cdot\msgp{x_{kt}}{\Psi_{z_{l,kt}}}(x_{kt})\nonumber\\
    &\qquad\cdot\msgp{\lvec{z}_{l,kt}}{\Psi_{z_{l,kt}}}(\lvec{z}_{l,kt})\cdot\msgp{\lvec{g}_{l,k}}{\Psi_{z_{l,kt}}}(\lvec{g}_{l,k})\,d\lvec{z}_{l,kt}\nonumber\\
    &\quad= \sum_{x_{kt}}\msgp{x_{kt}}{\Psi_{z_{l,kt}}}\!(x_{kt})\!\cdot\!\msgp{\Psi_{y_{l,t}}}{\lvec{z}_{l,kt}}\!(\lvec{g}_{l,k}x_{kt})\!\cdot\!\msgp{\lvec{g}_{l,k}}{\Psi_{z_{l,kt}}}\!(\lvec{g}_{l,k}),
    \label{eq:app_mmd_Psi_z_g_0}
\end{align}
where~\eqref{eq:app_mmd_Psi_z_g_0} is obtained by the sifting property of the Dirac delta function and using $\msgp{\lvec{z}_{l,kt}}{\Psi_{z_{l,kt}}}(\lvec{z}_{l,kt}) = \msgp{\Psi_{y_{l,t}}}{\lvec{z}_{l,kt}}(\lvec{z}_{l,kt})$.
The local \ac{BG} belief is further computed by considering the two parts modeling activity and inactivity of the \ac{BG} distribution separately.
For the \ac{BG} inactive component, only the Dirac at $\lvec{g}_{l,k}=\lvec{0}$ contributes, i.e., $\msgp{\lvec{g}_{l,k}}{\Psi_{z_{l,kt}}}(\lvec{g}_{l,k}=\lvec{0})={1-\Actmsg{\lvec{g}_{l,k}}{\Psi_{z_{l,kt}}}}$.
Thus,
\begin{align}
    &\mmd{\Psi_{z_{l,kt}}}{\lvec{g}_{l,k}}(\lvec{g}_{l,k}=\lvec{0})\nonumber\\
    &\quad\propto \sum_{x_{kt}}\msgp{x_{kt}}{\Psi_{z_{l,kt}}}(x_{kt})\cdot\theta(0)\cdot (1-\Actmsg{\lvec{g}_{l,k}}{\Psi_{z_{l,kt}}})\nonumber\\
    &\quad= \theta(0)\cdot (1-\Actmsg{\lvec{g}_{l,k}}{\Psi_{z_{l,kt}}}),
    \label{eq:app_mmd_Psi_z_g_inact}
\end{align}
with $\msgp{\Psi_{y_{l,t}}}{\lvec{z}_{l,kt}}(\lvec{0})=\theta(0)$ and $\theta(x)$ defined in~\eqref{eq:theta_mmd_Psi_z_z}.
For the \ac{BG} active component, only the Gaussian part is considered, i.e., $\msgp{\lvec{g}_{l,k}}{\Psi_{z_{l,kt}}}(\lvec{g}_{l,k})=\Actmsg{\lvec{g}_{l,k}}{\Psi_{z_{l,kt}}}\cdot\mathcal{CN}\big(\lvec{g}_{l,k}|\Mumsg{\lvec{g}_{l,k}}{\Psi_{z_{l,kt}}},\Cmsg{\lvec{g}_{l,k}}{\Psi_{z_{l,kt}}}\big)$,
\begin{align}
    &\mmd{\Psi_{z_{l,kt}}}{\lvec{g}_{l,k}}(\lvec{g}_{l,k}\neq\lvec{0})\nonumber\\
    &\quad\propto \sum_{x_{kt}}\msgp{x_{kt}}{\Psi_{z_{l,kt}}}\!(x_{kt})\!\cdot\!\mathcal{CN}\big(\lvec{g}_{l,k}x_{kt}|\Mumsg{\Psi_{y_{l,t}}}{\lvec{z}_{l,kt}}\!,\Cmsg{\Psi_{y_{l,t}}}{\lvec{z}_{l,kt}}\big)\nonumber\\
    &\qquad\cdot\Actmsg{\lvec{g}_{l,k}}{\Psi_{z_{l,kt}}}\cdot\mathcal{CN}\big(\lvec{g}_{l,k}|\Mumsg{\lvec{g}_{l,k}}{\Psi_{z_{l,kt}}},\Cmsg{\lvec{g}_{l,k}}{\Psi_{z_{l,kt}}}\big)\nonumber\\
    \begin{split}
    &\quad= \sum_{x_{kt}}\Actmsg{\lvec{g}_{l,k}}{\Psi_{z_{l,kt}}}\cdot\phi(x_{kt})\\
    &\qquad\cdot\mathcal{CN}\big(\lvec{g}_{l,k}|\Mumsga{\lvec{z}_{l,kt}}{}\!(x_{kt})x_{kt}^{-1},\Cmsga{\lvec{z}_{l,kt}}{}\!(x_{kt})|x_{kt}|^{-2}\big),
	\label{eq:app_mmd_Psi_z_g_act}
    \end{split}
\end{align}
which is obtained by applying the Gaussian scaling and product lemma with  $\phi(x_{kt})$, $\Mumsga{\lvec{z}_{l,kt}}{}(x_{kt})$, and $\Cmsga{\lvec{z}_{l,kt}}{}(x_{kt})$ defined in~\eqref{eq:phi_mmd_Psi_z},~\eqref{eq:mu_tmp_mmd_Psi_z_z}, and~\eqref{eq:C_tmp_mmd_Psi_z_z}, respectively.
Thus, the active component of $\mmd{\Psi_{z_{l,kt}}}{\lvec{g}_{l,k}}(\lvec{g}_{l,k})$ in~\eqref{eq:app_mmd_Psi_z_g_act} is a Gaussian mixture whose components correspond to the transmit symbols $x_{kt}$.
The computation of the expectation of the sufficient statistics of this mixture distribution for moment matching may yield a mean close to zero, which would indicate activity with a weak channel.
To prevent this issue which determines a high number of false alarms, we select the mixture component with the largest weight, denoted by $x^*$ in~\eqref{eq:x_star_mmd_Psi_z_g}.
Note that in the pilot phase for $t<T_p$, the transmitted symbol $x_{kt}=\pilot{x}_{kt}$ is known a priori and, hence, $x^*=\pilot{x}_{kt}$.
With this pragmatic approximation, we obtain
\begin{equation}
    \mmd{\Psi_{z_{l,kt}}}{\lvec{g}_{l,k}}(\lvec{g}_{l,k}) \approx \mathcal{BG}\big(\lvec{g}_{l,k}|\Actmsgb{\lvec{g}_{l,kt}}{},\Mumsgb{\lvec{g}_{l,kt}}{},\Cmsgb{\lvec{g}_{l,kt}}{}\big) 
    \label{eq:app_m_Psi_z_g}
\end{equation}
with parameters $\Actmsgb{\lvec{g}_{l,kt}}{}$, $\Mumsgb{\lvec{g}_{l,kt}}{}$, and $\Cmsgb{\lvec{g}_{l,kt}}{}$ given by~\eqref{eq:act_mmd_Psi_z_g},~\eqref{eq:gamma_mmd_Psi_z_g}, and~\eqref{eq:Lambda_mmd_Psi_z_g}, respectively.
By applying the \ac{EP} update rule~\eqref{eq:EP_fac_to_var_message}, the final message is defined by the parameters in~\eqref{eq:kappa_Psi_z_g},~\eqref{eq:mu_Psi_z_g}, and~\eqref{eq:C_Psi_z_g}.

\vspace*{-2mm}
\subsection{Message Update for $\msg{\Psi_{g_{l,k}}}{u_k}$}\label{subsec:MP_Psi_g_u}
The local belief of $u_k$ at the factor node $\Psi_{g_{l,k}}$ is given by
\begin{align}
	&\mmd{\Psi_{g_{l,k}}}{u_k}(u_k)\nonumber\\
    &\quad\propto \int\!\!\int\delta(\lvec{g}_{l,k}-\lvec{h}_{l,k}u_k)\cdot\msgp{u_k}{\Psi_{g_{l,k}}}(u_k)\cdot\msgp{\lvec{g}_{l,k}}{\Psi_{g_{l,k}}}(\lvec{g}_{l,k})\nonumber\\
    &\qquad\cdot\msgp{\lvec{h}_{l,k}}{\Psi_{g_{l,k}}}(\lvec{h}_{l,k})\,d\lvec{g}_{l,k}\,d\lvec{h}_{l,k}\nonumber\\
    \begin{split}
	&\quad= \int\msgp{u_k}{\Psi_{g_{l,k}}}(u_k)\cdot\msgp{\lvec{g}_{l,k}}{\Psi_{g_{l,k}}}(\lvec{h}_{l,k}u_k)\\
    &\qquad\cdot\msgp{\Psi_{h_{l,k}}}{\lvec{h}_{l,k}}(\lvec{h}_{l,k})\,d\lvec{h}_{l,k}
    \end{split}
	\label{eq:app_mmd_Psi_g_u_0}
\end{align}
which is obtained by the sifting property of the Dirac delta function and using $\msgp{\lvec{h}_{l,k}}{\Psi_{g_{l,k}}}(\lvec{h}_{l,k})=\msgp{\Psi_{h_{l,k}}}{\lvec{h}_{l,k}}(\lvec{h}_{l,k})$.
The local categorical belief is further computed by treating separately the two components modeling the Bernoulli distribution.
For the inactive case, i.e., $u_k=0$, only the Dirac component of the \ac{BG} distribution $\msgp{\lvec{g}_{l,k}}{\Psi_{g_{l,k}}}(\lvec{g}_{l,k})$ contributes, i.e., $\msgp{\lvec{g}_{l,k}}{\Psi_{g_{l,k}}}(\lvec{h}_{l,k}u_k=\lvec{0})=1-\Actmsg{\lvec{g}_{l,k}}{\Psi_{g_{l,k}}}$
\begin{align}
    &\mmd{\Psi_{g_{l,k}}}{u_k}(u_k=0)\nonumber\\
    &\quad\propto \int\msgp{u_k}{\Psi_{g_{l,k}}}(0)\cdot(1-\Actmsg{\lvec{g}_{l,k}}{\Psi_{g_{l,k}}})\cdot\msgp{\Psi_{h_{l,k}}}{\lvec{h}_{l,k}}(\lvec{h}_{l,k})\,d\lvec{h}_{l,k}\nonumber\\
    &\quad= \msgp{u_k}{\Psi_{g_{l,k}}}(0)\cdot(1-\Actmsg{\lvec{g}_{l,k}}{\Psi_{g_{l,k}}}).
    \label{eq:app_mmd_Psi_g_u_inact}
\end{align}
For the active case, i.e., $u_k=1$, only the Gaussian component of the \ac{BG} distribution $\msgp{\lvec{g}_{l,k}}{\Psi_{g_{l,k}}}(\lvec{g}_{l,k})$ is relevant, i.e., $\msgp{\lvec{g}_{l,k}}{\Psi_{g_{l,k}}}(\lvec{g}_{l,k})=\msgp{\lvec{g}_{l,k}}{\Psi_{g_{l,k}}}(\lvec{h}_{l,k})=\Actmsg{\lvec{g}_{l,k}}{\Psi_{g_{l,k}}}\cdot\mathcal{CN}\big(\lvec{h}_{l,k}|\Mumsg{\lvec{g}_{l,k}}{\Psi_{g_{l,k}}},\Cmsg{\lvec{g}_{l,k}}{\Psi_{g_{l,k}}}\big)$,
\begin{align}
    &\mmd{\Psi_{g_{l,k}}}{u_k}(u_k=1)\\
    &\quad\propto \int\msgp{u_k}{\Psi_{g_{l,k}}}(1)\cdot\Actmsg{\lvec{g}_{l,k}}{\Psi_{g_{l,k}}}\nonumber\\
    &\qquad\cdot\mathcal{CN}\big(\lvec{h}_{l,k}|\Mumsg{\lvec{g}_{l,k}}{\Psi_{g_{l,k}}}\!,\Cmsg{\lvec{g}_{l,k}}{\Psi_{g_{l,k}}}\!\big)\!\cdot\!\msgp{\Psi_{h_{l,k}}}{\lvec{h}_{l,k}}(\lvec{h}_{l,k})\,d\lvec{h}_{l,k}\nonumber\\
    &\quad= \msgp{u_k}{\Psi_{g_{l,k}}}(1)\cdot\Actmsg{\lvec{g}_{l,k}}{\Psi_{g_{l,k}}}\cdot\vartheta(1),
    \label{eq:app_mmd_Psi_g_u_act}
\end{align}
which is obtained by applying the Gaussian product lemma with $\vartheta(u_k)$ defined in~\eqref{eq:theta_mmd_Psi_g_u}.
Then, combining the active and inactive components, we obtain
\begin{equation}
	\mmd{\Psi_{g_{l,k}}}{u_k}(u_k) \!\propto\! \msgp{u_k}{\Psi_{g_{l,k}}}(u_k)\cdot\begin{cases}
	1-\Actmsg{\lvec{g}_{l,k}}{\Psi_{g_{l,k}}} & \!\!\text{for }u_k=0\\
	\Actmsg{\lvec{g}_{l,k}}{\Psi_{g_{l,k}}}\cdot\vartheta(1) & \!\!\text{for }u_k=1
	\end{cases}.
	\label{eq:app_mmd_Psi_g_u}
\end{equation}
The projection in~\eqref{eq:EP_fac_to_var_message} is superfluous since $\mmd{\Psi_{g_{l,k}}}{u_k}(u_k)$ is already a categorical distribution.
Hence, the denominator of~\eqref{eq:EP_fac_to_var_message} cancels with the first term in~\eqref{eq:app_mmd_Psi_g_u}.
Thus, the final message update rule is given by~\eqref{eq:m_Psi_g_u}.

\vspace*{-2mm}
\subsection{Message Update for $\msg{\Psi_{g_{l,k}}}{\lvec{g}_{l,k}}$}\label{subsec:MP_Psi_g_g}
The local belief of $\lvec{g}_{l,k}$ at the factor node $\Psi_{g_{l,k}}$ is given by
\begin{align}
	&\mmd{\Psi_{g_{l,k}}}{\lvec{g}_{l,k}}(\lvec{g}_{l,k})\nonumber\\
    &\quad\propto \sum_{u_k}\int\delta(\lvec{g}_{l,k}-\lvec{h}_{l,k}u_k)\cdot\msgp{u_k}{\Psi_{g_{l,k}}}(u_k)\cdot\msgp{\lvec{g}_{l,k}}{\Psi_{g_{l,k}}}(\lvec{g}_{l,k})\nonumber\\
    &\qquad\cdot\msgp{\lvec{h}_{l,k}}{\Psi_{g_{l,k}}}(\lvec{h}_{l,k})\,d\lvec{h}_{l,k}\nonumber\\
	&\quad\overset{(a)}{=} \msgp{\lvec{g}_{l,k}}{\Psi_{g_{l,k}}}(\lvec{g}_{l,k})\cdot\big(\msgp{u_k}{\Psi_{g_{l,k}}}(0)\cdot\delta(\lvec{g}_{l,k})\nonumber\\
    &\qquad+\msgp{u_k}{\Psi_{g_{l,k}}}(1)\cdot\msgp{\Psi_{h_{l,k}}}{\lvec{h}_{l,k}}(\lvec{g}_{l,k})\big)\nonumber\\
\begin{split}
	&\quad= \msgp{\lvec{g}_{l,k}}{\Psi_{g_{l,k}}}(\lvec{g}_{l,k})\cdot\mathcal{BG}\big(\lvec{g}_{l,k}|\msgp{u_k}{\Psi_{g_{l,k}}}(1),\tilde{\gvec{\mu}}_{h_{l,k}},\tilde{\lmat{C}}_{h_{l,k}}\big)
	\label{eq:app_mmd_Psi_g_g}
\end{split}
\end{align}
where $(a)$ is obtained by explicitly evaluating the sum over $u_k$ and utilizing the sifting property of the Dirac delta function along with the identity $\msgp{\lvec{h}_{l,k}}{\Psi_{g_{l,k}}}(\lvec{h}_{l,k})=\msgp{\Psi_{h_{l,k}}}{\lvec{h}_{l,k}}(\lvec{h}_{l,k})$.
The final equation~\eqref{eq:app_mmd_Psi_g_g} is obtained by using $\msgp{\Psi_{h_{l,k}}}{\lvec{h}_{l,k}}(\lvec{h}_{l,k})=\mathcal{CN}\big(\lvec{h}_{l,k}|\tilde{\gvec{\mu}}_{h_{l,k}},\tilde{\lmat{C}}_{h_{l,k}}\big)$ and noting that the second factor in $(a)$ written in brackets is a \ac{BG} distribution.
The projection in~\eqref{eq:EP_fac_to_var_message} is superfluous since $\mmd{\Psi_{g_{l,k}}}{\lvec{g}_{l,k}}(\lvec{g}_{l,k})$ is the product of two \ac{BG} distributions which is again a \ac{BG} distribution.
Hence, the denominator of~\eqref{eq:EP_fac_to_var_message} cancels with the first term in~\eqref{eq:app_mmd_Psi_g_g}.
Thus, the final message parameters are given by~\eqref{eq:act_Psi_g_g},~\eqref{eq:mu_Psi_g_g}, and~\eqref{eq:C_Psi_g_g}.

\vspace*{-2mm}
\subsection{Message Update for $\msg{\Psi_{z_{l,kt}}}{\lvec{z}_{l,kt}}$}\label{subsec:m_Psi_z_z}
The local belief of $\lvec{z}_{l,kt}$ at the factor node $\Psi_{z_{l,kt}}$ is given by
\begin{align}
	&\mmd{\Psi_{z_{l,kt}}}{\lvec{z}_{l,kt}}(\lvec{z}_{l,kt})\nonumber\\
	&\quad\propto \sum_{x_{kt}}\int\delta(\lvec{z}_{l,kt}-\lvec{g}_{l,k}x_{kt})\cdot\msgp{x_{kt}}{\Psi_{z_{l,kt}}}(x_{kt})\nonumber\\
    &\qquad\cdot\msgp{\lvec{z}_{l,kt}}{\Psi_{z_{l,kt}}}(\lvec{z}_{l,kt})\nonumber\cdot\msgp{\lvec{g}_{l,k}}{\Psi_{z_{l,kt}}}(\lvec{g}_{l,k})\,d\lvec{g}_{l,k}\nonumber\\
	&\quad\overset{(a)}{=} \sum_{x_{kt}}|x_{kt}|^{-2N}\cdot\msgp{x_{kt}}{\Psi_{z_{l,kt}}}(x_{kt})\cdot\msgp{\Psi_{y_{l,t}}}{\lvec{z}_{l,kt}}(\lvec{z}_{l,kt})\nonumber\\
    &\qquad\cdot\msgp{\lvec{g}_{l,k}}{\Psi_{z_{l,kt}}}\Big(\frac{\lvec{z}_{l,kt}}{x_{kt}}\Big)\nonumber\\
	&\quad\overset{(b)}{\approx} \sum_{x_{kt}}|x_{kt}|^{-2N}\cdot\msgp{x_{kt}}{\Psi_{z_{l,kt}}}(x_{kt})\cdot\msgp{\Psi_{y_{l,t}}}{\lvec{z}_{l,kt}}(\lvec{z}_{l,kt})\nonumber\\
    &\qquad\cdot\Actmsg{\lvec{g}_{l,k}}{\Psi_{z_{l,kt}}}\cdot\mathcal{CN}\Big(\frac{\lvec{z}_{l,kt}}{x_{kt}}\Big|\mumsg{\lvec{g}_{l,k}}{\Psi_{z_{l,kt}}},\Cmsg{\lvec{g}_{l,k}}{\Psi_{z_{l,kt}}}\Big)\nonumber\\
	&\quad= \sum_{x_{kt}}\catmsg{x_{kt}}{\Psi_{z_{l,kt}}}(x_{kt})\cdot\theta(x_{kt})\cdot\Actmsg{\lvec{g}_{l,k}}{\Psi_{z_{l,kt}}}\nonumber\\
    &\qquad\cdot\mathcal{CN}\big(\lvec{z}_{l,kt}|\Mumsga{\lvec{z}_{l,kt}}{}(x_{kt}),\Cmsga{\lvec{z}_{l,kt}}{}(x_{kt})\big),
	\label{eq:app_mmd_Psi_z_z}
\end{align}
where $(a)$ is obtained by applying the scaling and sifting property of the Dirac delta function~\cite{Candan2021} and using $\msgp{\lvec{z}_{l,kt}}{\Psi_{z_{l,kt}}}(\lvec{z}_{l,kt}) = \msgp{\Psi_{y_{l,t}}}{\lvec{z}_{l,kt}}(\lvec{z}_{l,kt})$, and $(b)$ is obtained by approximating the \ac{BG} distribution $\msgp{\lvec{g}_{l,k}}{\Psi_{z_{l,kt}}}$ by its Gaussian part.
The final result in~\eqref{eq:app_mmd_Psi_z_z} is obtained by applying the Gaussian scaling and product lemma with $\theta(x_{kt})$, $\Mumsga{\lvec{z}_{l,kt}}{}(x_{kt})$, and $\Cmsga{\lvec{z}_{l,kt}}{}(x_{kt})$ given by~\eqref{eq:theta_mmd_Psi_z_z},~\eqref{eq:mu_tmp_mmd_Psi_z_z}, and~\eqref{eq:C_tmp_mmd_Psi_z_z}, respectively.
The normalization constant for $\mmd{\Psi_{z_{l,kt}}}{\lvec{z}_{l,kt}}(\lvec{z}_{l,kt})$ is given by $\tilde{Z}_{\Psi_{z_{l,kt}}}\!=\Actmsg{\lvec{g}_{l,k}}{\Psi_{z_{l,kt}}}\!\!\cdot Z_{\Psi_{z_{l,kt}}}$ where $Z_{\Psi_{z_{l,kt}}}$ is defined in~\eqref{eq:Z_Psi_z}.
According to~\eqref{eq:app_mmd_Psi_z_z}, $\mmd{\Psi_{z_{l,kt}}}{\lvec{z}_{l,kt}}(\lvec{z}_{l,kt})$ is a Gaussian mixture with mean vector and covariance matrix given by~\eqref{eq:mu_mmd_Psi_z_z} and~\eqref{eq:C_mmd_Psi_z_z}, respectively.
Note that in the pilot phase for $t<T_p$, the transmitted symbol $x_{kt}=\pilot{x}_{kt}$ is known a priori.
Hence, $\mmd{\Psi_{z_{l,kt}}}{\lvec{z}_{l,kt}}(\lvec{z}_{l,kt})$ reduces to a Gaussian distribution for $t<T_p$ with mean vector $\Mumsgb{\lvec{z}_{l,kt}}{}=\Mumsga{\lvec{z}_{l,kt}}{}(\pilot{x}_{kt})$ and covariance matrix $\Cmsgb{\lvec{z}_{l,kt}}{}=\Cmsga{\lvec{z}_{l,kt}}{}(\pilot{x}_{kt})$.
The final message parameters are given by~\eqref{eq:mu_Psi_z_z} and~\eqref{eq:C_Psi_z_z}, which are computed according to~\eqref{eq:EP_fac_to_var_message}.

\linespread{0.88}
\vspace*{-2mm}
\bibliographystyle{IEEEtran}
\bibliography{IEEEabrv,references_abbrev}

\begin{thebibliography}{10}
\providecommand{\url}[1]{#1}
\csname url@samestyle\endcsname
\providecommand{\newblock}{\relax}
\providecommand{\bibinfo}[2]{#2}
\providecommand{\BIBentrySTDinterwordspacing}{\spaceskip=0pt\relax}
\providecommand{\BIBentryALTinterwordstretchfactor}{4}
\providecommand{\BIBentryALTinterwordspacing}{\spaceskip=\fontdimen2\font plus
\BIBentryALTinterwordstretchfactor\fontdimen3\font minus \fontdimen4\font\relax}
\providecommand{\BIBforeignlanguage}[2]{{%
\expandafter\ifx\csname l@#1\endcsname\relax
\typeout{** WARNING: IEEEtran.bst: No hyphenation pattern has been}%
\typeout{** loaded for the language `#1'. Using the pattern for}%
\typeout{** the default language instead.}%
\else
\language=\csname l@#1\endcsname
\fi
#2}}
\providecommand{\BIBdecl}{\relax}
\BIBdecl

\bibitem{Forsch2024}
C.~Forsch, A.~Karataev, and L.~Cottatellucci, ``Distributed joint user activity detection, channel estimation, and data detection via expectation propagation in cell-free massive {MIMO},'' in \emph{Proc. IEEE 25th Int. Workshop Signal Process. Advances Wireless Commun. (SPAWC)}, 2024, pp. 531--535, extended version available at: \url{https://arxiv.org/abs/2405.09914}.

\bibitem{Mahmood2021}
N.~H. Mahmood \emph{et~al.}, ``Machine type communications: key drivers and enablers towards the 6{G} era,'' \emph{EURASIP Journal on Wireless Communications and Networking}, vol. 2021, no.~1, p. 134, 2021.

\bibitem{Liu2018b}
L.~Liu, E.~G. Larsson, W.~Yu, P.~Popovski, C.~Stefanovic, and E.~de~Carvalho, ``Sparse signal processing for grant-free massive connectivity: A future paradigm for random access protocols in the internet of things,'' \emph{{IEEE} Signal Process. Mag.}, vol.~35, no.~5, pp. 88--99, 2018.

\bibitem{Shahab2020}
M.~B. Shahab, R.~Abbas, M.~Shirvanimoghaddam, and S.~J. Johnson, ``Grant-free non-orthogonal multiple access for {IoT}: A survey,'' \emph{{IEEE} Commun. Surveys Tuts.}, vol.~22, no.~3, pp. 1805--1838, 2020.

\bibitem{Gao2024}
Z.~Gao \emph{et~al.}, ``Compressive-sensing-based grant-free massive access for 6{G} massive communication,'' \emph{{IEEE} Internet Things J.}, vol.~11, no.~5, pp. 7411--7435, 2024.

\bibitem{Ngo2017}
H.~Q. Ngo, A.~Ashikhmin, H.~Yang, E.~G. Larsson, and T.~L. Marzetta, ``Cell-free massive {MIMO} versus small cells,'' \emph{{IEEE} Trans. Wireless Commun.}, vol.~16, no.~3, pp. 1834--1850, 2017.

\bibitem{Ngo2018}
H.~Q. Ngo, L.-N. Tran, T.~Q. Duong, M.~Matthaiou, and E.~G. Larsson, ``On the total energy efficiency of cell-free massive {MIMO},'' \emph{{IEEE} Trans. Green Commun. Netw.}, vol.~2, no.~1, pp. 25--39, 2018.

\bibitem{Mohammadi2024}
M.~Mohammadi, Z.~Mobini, H.~Quoc~Ngo, and M.~Matthaiou, ``Next-generation multiple access with cell-free massive {MIMO},'' \emph{Proc. {IEEE}}, vol. 112, no.~9, pp. 1372--1420, 2024.

\bibitem{Wang2021}
H.~Wang, J.~Wang, and J.~Fang, ``Grant-free massive connectivity in massive {MIMO} systems: Collocated versus cell-free,'' \emph{{IEEE} Wireless Commun. Lett.}, vol.~10, no.~3, pp. 634--638, 2021.

\bibitem{Ganesan2021}
U.~K. Ganesan, E.~Björnson, and E.~G. Larsson, ``Clustering-based activity detection algorithms for grant-free random access in cell-free massive {MIMO},'' \emph{{IEEE} Trans. Commun.}, vol.~69, no.~11, pp. 7520--7530, 2021.

\bibitem{Xu2025}
Y.~Xu, E.~G. Larsson, E.~A. Jorswieck, X.~Li, S.~Jin, and T.-H. Chang, ``Distributed signal processing for extremely large-scale antenna array systems: State-of-the-art and future directions,'' \emph{{IEEE} J. Sel. Topics Signal Process.}, vol.~19, no.~2, pp. 304--330, 2025.

\bibitem{Liu2018a}
L.~Liu and W.~Yu, ``Massive connectivity with massive {MIMO}—{P}art {I}: Device activity detection and channel estimation,'' \emph{{IEEE} Trans. Signal Process.}, vol.~66, no.~11, pp. 2933--2946, 2018.

\bibitem{Ke2020}
M.~Ke, Z.~Gao, Y.~Wu, X.~Gao, and R.~Schober, ``Compressive sensing-based adaptive active user detection and channel estimation: Massive access meets massive {MIMO},'' \emph{{IEEE} Trans. Signal Process.}, vol.~68, pp. 764--779, 2020.

\bibitem{Ngo2012}
H.~Q. Ngo and E.~G. Larsson, ``{EVD}-based channel estimation in multicell multiuser {MIMO} systems with very large antenna arrays,'' in \emph{Proc. IEEE Int. Conf. Acoust., Speech, Signal Process. (ICASSP)}, 2012.

\bibitem{Yin2013}
H.~Yin, D.~Gesbert, M.~Filippou, and Y.~Liu, ``A coordinated approach to channel estimation in large-scale multiple-antenna systems,'' \emph{{IEEE} J. Sel. Areas Commun.}, vol.~31, no.~2, pp. 264--273, 2013.

\bibitem{Cottatellucci2013}
L.~Cottatellucci, R.~R. M{\"u}ller, and M.~Vehkapera, ``Analysis of pilot decontamination based on power control,'' in \emph{Proc. IEEE 77th Veh. Technol. Conf. (VTC-Spring)}, 2013.

\bibitem{Mueller2014}
R.~R. M{\"u}ller, L.~Cottatellucci, and M.~Vehkapera, ``Blind pilot decontamination,'' \emph{{IEEE} J. Sel. Topics Signal Process.}, vol.~8, no.~5, pp. 773--786, 2014.

\bibitem{Yin2016}
H.~Yin, L.~Cottatellucci, D.~Gesbert, R.~R. M{\"u}ller, and G.~He, ``Robust pilot decontamination based on joint angle and power domain discrimination,'' \emph{{IEEE} Trans. Signal Process.}, vol.~64, no.~11, pp. 2990--3003, 2016.

\bibitem{Yin2014}
H.~Yin, D.~Gesbert, and L.~Cottatellucci, ``Dealing with interference in distributed large-scale {MIMO} systems: A statistical approach,'' \emph{{IEEE} J. Sel. Topics Signal Process.}, vol.~8, no.~5, pp. 942--953, 2014.

\bibitem{Chen2018}
Z.~Chen and E.~Bj{\"o}rnson, ``Channel hardening and favorable propagation in cell-free massive {MIMO} with stochastic geometry,'' \emph{{IEEE} Trans. Commun.}, vol.~66, no.~11, pp. 5205--5219, 2018.

\bibitem{Gholami2020a}
R.~Gholami, L.~Cottatellucci, and D.~Slock, ``Favorable propagation and linear multiuser detection for distributed antenna systems,'' in \emph{Proc. IEEE Int. Conf. Acoust., Speech, Signal Process. (ICASSP)}, 2020.

\bibitem{Gholami2020b}
------, ``Channel models, favorable propagation and {MultiStage} linear detection in cell-free massive {MIMO},'' in \emph{Proc. IEEE Int. Symp. Inf. Theory (ISIT)}, 2020.

\bibitem{Bjoernson2018}
E.~Bj{\"o}rnson, J.~Hoydis, and L.~Sanguinetti, ``Massive {MIMO} has unlimited capacity,'' \emph{{IEEE} Trans. Wireless Commun.}, vol.~17, no.~1, pp. 574--590, 2018.

\bibitem{Bjoernson2020}
E.~Bj{\"o}rnson and L.~Sanguinetti, ``Making cell-free massive {MIMO} competitive with {MMSE} processing and centralized implementation,'' \emph{{IEEE} Trans. Wireless Commun.}, vol.~19, no.~1, pp. 77--90, 2020.

\bibitem{Polegre2021}
A.~{\'{A}}. Polegre, L.~Sanguinetti, and A.~G. Armada, ``Pilot decontamination processing in cell-free massive {MIMO},'' \emph{{IEEE} Commun. Lett.}, vol.~25, no.~12, pp. 3990--3994, 2021.

\bibitem{Gholami2021a}
R.~Gholami, L.~Cottatellucci, and D.~Slock, ``Tackling pilot contamination in cell-free massive {MIMO} by joint channel estimation and linear multi-user detection,'' in \emph{Proc. IEEE Int. Symp. Inf. Theory (ISIT)}, 2021.

\bibitem{Song2022}
H.~Song, T.~Goldstein, X.~You, C.~Zhang, O.~Tirkkonen, and C.~Studer, ``Joint channel estimation and data detection in cell-free massive {MU}-{MIMO} systems,'' \emph{{IEEE} Trans. Wireless Commun.}, vol.~21, no.~6, pp. 4068--4084, 2022.

\bibitem{Karataev2024}
A.~Karataev, C.~Forsch, and L.~Cottatellucci, ``Bilinear expectation propagation for distributed semi-blind joint channel estimation and data detection in cell-free massive {MIMO},'' \emph{{IEEE} Open J. Signal Process.}, vol.~5, pp. 284--293, 2024.

\bibitem{Forsch2025}
C.~Forsch, Z.~Zhao, D.~Slock, and L.~Cottatellucci, ``Bayesian learning for pilot decontamination in cell-free massive {MIMO},'' in \emph{Proc. 28th Int. Workshop Smart Antennas (WSA)}, 2025.

\bibitem{Zhao2024}
Z.~Zhao and D.~Slock, ``Decentralized message-passing for semi-blind channel estimation in cell-free systems based on {Bethe} free energy optimization,'' in \emph{Proc. 58th Asilomar Conf. Signals, Syst., and Comput.}, 2024, pp. 1443--1447.

\bibitem{Jiang2020}
S.~Jiang, X.~Yuan, X.~Wang, C.~Xu, and W.~Yu, ``Joint user identification, channel estimation, and signal detection for grant-free {NOMA},'' \emph{{IEEE} Trans. Wireless Commun.}, vol.~19, no.~10, pp. 6960--6976, 2020.

\bibitem{Zhang2020}
Y.~Zhang, Z.~Yuan, Q.~Guo, Z.~Wang, J.~Xi, and Y.~Li, ``Bayesian receiver design for grant-free {NOMA} with message passing based structured signal estimation,'' \emph{IEEE Transactions on Vehicular Technology}, vol.~69, no.~8, pp. 8643--8656, 2020.

\bibitem{Zhang2025}
C.~Zhang, Y.~Liu, J.~Hu, and K.~Yang, ``Joint user identification, channel estimation, and data detection for grant-free {NOMA} in {LEO} satellite communications,'' \emph{{IEEE} J. Sel. Areas Commun.}, vol.~43, no.~1, pp. 107--121, 2025.

\bibitem{Zou2020}
Q.~Zou, H.~Zhang, D.~Cai, and H.~Yang, ``A low-complexity joint user activity, channel and data estimation for grant-free massive {MIMO} systems,'' \emph{{IEEE} Signal Processing Lett.}, vol.~27, pp. 1290--1294, 2020.

\bibitem{Zhang2023}
S.~Zhang, Y.~Cui, and W.~Chen, ``Joint device activity detection, channel estimation and signal detection for massive grant-free access via {BiGAMP},'' \emph{{IEEE} Trans. Signal Process.}, vol.~71, pp. 1200--1215, 2023.

\bibitem{Bian2023}
X.~Bian, Y.~Mao, and J.~Zhang, ``Joint activity detection, channel estimation, and data decoding for grant-free massive random access,'' \emph{{IEEE} Internet Things J.}, vol.~10, no.~16, pp. 14\,042--14\,057, 2023.

\bibitem{Iimori2021}
H.~Iimori, T.~Takahashi, K.~Ishibashi, G.~T.~F. de~Abreu, and W.~Yu, ``Grant-free access via bilinear inference for cell-free {MIMO} with low-coherence pilots,'' \emph{{IEEE} Trans. Wireless Commun.}, vol.~20, no.~11, pp. 7694--7710, 2021.

\bibitem{Sun2025}
G.~Sun, M.~Cao, W.~Wang, W.~Xu, and C.~Studer, ``Deep-unfolded massive grant-free transmission in cell-free wireless communication systems,'' \emph{{IEEE} Trans. Signal Process.}, vol.~73, pp. 1094--1109, 2025.

\bibitem{Minka2001a}
T.~P. Minka, ``A family of algorithms for approximate {B}ayesian inference,'' Ph.D. dissertation, Massachusetts Inst. Technol., Cambridge, 2001.

\bibitem{Minka2001b}
------, ``Expectation propagation for approximate {B}ayesian inference,'' in \emph{Proc. 17th Conf. Uncertainty Artif. Intell. (UAI)}, 2001, pp. 362--369.

\bibitem{Wainwright2007}
M.~J. Wainwright and M.~I. Jordan, ``Graphical models, exponential families, and variational inference,'' \emph{Found. Trends{\textregistered} Mach. Learn.}, vol.~1, no. 1{\textendash}2, pp. 1--305, 2007.

\bibitem{Bromiley2003}
P.~A. Bromiley, ``Products and convolutions of {G}aussian distributions,'' Tech. Rep. Tina Memo No. 2003-003, 2003.

\bibitem{Ngo2020}
K.-H. Ngo, M.~Guillaud, A.~Decurninge, S.~Yang, and P.~Schniter, ``Multi-user detection based on expectation propagation for the non-coherent {SIMO} multiple access channel,'' \emph{{IEEE} Trans. Wireless Commun.}, vol.~19, no.~9, pp. 6145--6161, 2020.

\bibitem{Vila2011}
J.~Vila and P.~Schniter, ``Expectation-maximization {B}ernoulli-{G}aussian approximate message passing,'' in \emph{Proc. 45th Asilomar Conf. Signals, Syst., Comput.}, 2011, pp. 799--803.

\bibitem{Lobato2013}
D.~Hern\'{a}ndez-Lobato, J.~M. Hern\'{a}ndez-Lobato, and P.~Dupont, ``Generalized spike-and-slab priors for {B}ayesian group feature selection using expectation propagation,'' \emph{J. Mach. Learn. Res.}, vol.~14, no.~1, p. 1891–1945, 2013.

\bibitem{Rusu2018}
C.~Rusu, N.~González-Prelcic, and R.~W. Heath, ``Algorithms for the construction of incoherent frames under various design constraints,'' \emph{Signal Processing}, vol. 152, pp. 363--372, 2018.

\bibitem{Demir2021}
{\"O}.~T. Demir, E.~Bj{\"o}rnson, and L.~Sanguinetti, ``Foundations of user-centric cell-free massive mimo,'' \emph{Found. Trends{\textregistered} Signal Process.}, vol.~14, no. 3-4, pp. 162--472, 2021.

\bibitem{Chen2017}
M.~Chen, Y.~Miao, Y.~Hao, and K.~Hwang, ``Narrow band internet of things,'' \emph{IEEE Access}, vol.~5, pp. 20\,557--20\,577, 2017.

\bibitem{Candan2021}
C.~Candan, ``Proper definition and handling of {D}irac delta functions [lecture notes],'' \emph{{IEEE} Signal Process. Mag.}, vol.~38, no.~3, pp. 186--203, 2021.

\bibitem{Papoulis2002}
A.~Papoulis and S.~U. Pillai, \emph{Probability, random variables and stochastic processes}, 4th~ed.\hskip 1em plus 0.5em minus 0.4em\relax McGraw-Hill, 2002.

\end{thebibliography}

\end{document}